\documentclass[journal,10pt,tbtags,twocolumn,twoside]{IEEEtran}

\usepackage[T1]{fontenc}
\usepackage{xcolor}
\usepackage{mathtools, amssymb, amsfonts, dsfont}
\usepackage{microtype}
\usepackage{url}

\usepackage[noadjust,sort,compress]{cite} 
\usepackage[standard,amsmath,thmmarks]{ntheorem}
\usepackage{xcolor}
\usepackage{booktabs}

\usepackage{graphicx}

\usepackage[caption=false,font=footnotesize]{subfig}
\newcommand{\HS}{\mathcal{H}}	% Hilbert space
\newcommand{\LS}{\mathcal{L}} % space of linear functions
\newcommand{\SBS}{\mathcal{S}} %  subspace S

\newcommand*\BR{\mathcal{R}}

\newcommand{\BA}{\mathbb{A}}
\newcommand{\BB}{\mathbb{B}}
\newcommand{\I}{\mathbb{I}}

\newcommand{\FA}{\mathbf{A}}   % Function A
\newcommand{\FB}{\mathbf{B}}   % Function B
\newcommand{\Fu}{\mathbf{u}}     % Function u
\newcommand{\Fx}{\mathbf{x}}     % Function x
\newcommand{\Fv}{\mathbf{v}}     % Function v
\newcommand{\Fz}{\mathbf{z}}     % Function v
\newcommand{\Ff}{\mathbf{f}}     % Function x
\newcommand{\FQ}{\mathbf{Q}}

\newcommand{\BI}{\mathbb{I}}     % Function x

\newcommand{\BP}{\mathbb{P}}     
\newcommand{\BQ}{\mathbb{Q}}
\newcommand{\ES}{\mathcal{W}} % Exact Space functions without taking cut distance zero.

\newcommand{\ESC}{\mathcal{W}_c}

\newcommand{\SA}{\mathbf{A^{[N]}}}
\newcommand{\SB}{\mathbf{B^{[N]}}}

\newcommand{\Sxt}{\mathbf{x^{[N]}_t}}

\newcommand{\DSxt}{\mathbf{\dot{x}^{[N]}_t}}

\newcommand{\Sut}{\mathbf{u^{[N]}_t}}
\newcommand{\Chi}{\mathds{1}}
\usepackage{tikz}
\tikzset{
  agent/.style={draw, fill=blue!25!white,circle, minimum size=2mm,
  font=\footnotesize},
  x = 1.5cm, y=1.5cm,
  every loop/.style={},
}

\newtheorem*{AppImp*}{Approximate Control Implementation}

\newcommand*\TRANS{{\mathpalette\doTRANS\empty}}
\makeatletter
\newcommand*\doTRANS[2]{\raisebox{\depth}{$\m@th#1\intercal$}}
\makeatother

\newcommand\MATRIX[1]{\begin{bmatrix} #1 \end{bmatrix}}

\begin{document}

\title{Subspace Decomposition for Graphon LQR: Applications to VLSNs of Harmonic Oscillators}
\author{\today}
 \author{Shuang~Gao,~\IEEEmembership{Member,~IEEE,} and Peter~E.~Caines, \IEEEmembership{Life Fellow, IEEE}% 
 \thanks{*The preliminary version of this work was presented at the 58th IEEE Conference on Decisions and Control, Nice, France, 2019.}
 \thanks{*This work is supported in part by NSERC (Canada), the U.S. ARL and ARO grant W911NF1910110, and the AFOSR Grant FA9550-19-1-0138.}
 \thanks{S. Gao, and Peter E. Caines are with the Department of Electrical and Computer Engineering, McGill University,
   Montreal, QC, Canada. 
   Emails: {\tt\small \{sgao, peterc\}@cim.mcgill.ca}.}%
 }

\maketitle

\begin{abstract}
Graphon control has been proposed and developed in \cite{ShuangPeterCDC17,ShuangPeterCDC19W2,ShuangPeterTAC18} to approximately solve control problems for very large-scale networks (VLSNs) of linear dynamical systems {based on graphon limits}. 
This paper provides a solution method based on invariant subspace decompositions for a class of graphon linear quadratic regulation (LQR) problems where the local dynamics share homogeneous parameters but the graphon couplings {may be} heterogeneous among the coupled agents. Graphon couplings in this paper appear in states, controls and costs, and these couplings may be represented by different graphons. 
By exploring a common invariant subspace of the couplings, the original problem is decomposed into a network coupled LQR problem of finite dimension and a decoupled infinite dimensional LQR problem.
A centralized optimal solution, and a nodal collaborative optimal control solution {where each agent computes its part of the optimal solution locally}, are established. {The application of these solutions to finite network LQR problems may be  via (i) the graphon control methodology \cite{ShuangPeterTAC18}, or (ii)  the representation of finite LQR problems as special cases of graphon LQR problems.} The complexity of these solutions involves solving one $nd\times nd$ dimensional Riccati equation and one $n\times n$ Riccati equation, where $n$ is the dimension of each nodal agent state and $d$ is the dimension of the nontrivial common invariant subspace of the coupling operators, {whereas a direct approach involves solving an $nN \times nN$ dimensional Riccati equation, where $N$ is the size of the network.}  
For situations where the graphon couplings do not  admit exact low-rank representations, approximate control is developed based on low-rank approximations.  Finally, an application to the regulation of harmonic oscillators coupled over large networks with uncertainties is demonstrated.    
\end{abstract}


\section{Introduction}

The study of very large-scale networks (VLSNs) of dynamical agents is motivated by systems such as smart grids, the Internet of Things (IoT), 5G communications, the spread of epidemics, very large-scale robotic networks and biological neuronal networks, among others.
Furthermore, research concerning the control of dynamical systems on complex networks typically involves the following:  controllability~\cite{liu2011controllability}, control energy~\cite{pasqualetti2014controllability}, input node selection~\cite{chen2017pinning},
 low-complexity control synthesis problems with simplified objective (e.g. consensus \cite{olfati2007consensus} or synchronization \cite{arenas2008synchronization}), simplified control (e.g. pinning control \cite{chen2017pinning} and ensemble control \cite{li2011ensemble}), low-rank (e.g. mean field) coupling \cite{yong2013linear,arabneydi2016team,zecevic2005global}, or patterned coupling \cite{hamilton2012patterned}.
However, the control of dynamical processes and agents on VLSNs still requires new theories, in particular those which generate scalable solutions. 

In a recent effort to solve control problems for very large-scale networks of linear dynamical systems, graphon control has been introduced to generate scalable approximate control solutions \cite{ShuangPeterCDC17,ShuangPeterTAC18,ShuangPeterCDC19W2}.  Dynamical systems coupled over networks of arbitrary sizes may be modelled by graphon dynamical control systems based on graphon theory \cite{borgs2008convergent,borgs2012convergent,lovasz2012large} and infinite dimensional linear system theory \cite{bensoussan2007representation,curtain1995introduction}. Under this representation a limit graphon control problem is  formulated based on the limit graphon (or an estimated graphon based upon given data) and an approximate solution to the original finite network control problem is then generated \cite{ShuangPeterTAC18}. {Moreover, this graphon approximate control based on the limit graphon control solution applies to a set of network systems of arbitrary sizes in the associate convergence sequence \cite{ShuangPeterTAC18}.}
Since a limit graphon system is  infinite dimensional {in terms of the number of agents}, an important issue in the graphon control methodology is the systematic generation of control laws for the corresponding infinite dimensional limit control problem. 

This article presents a study that provides solutions to a class of such problems based on invariant subspace decompositions, which generalizes the preliminary version based on eigen-decompositions in \cite{ShuangPeterCDC19W2}. %The underlying common finite dimensional invariant subspace structures of the couplings in the linear quadratic control problems allow low-complexity solutions. 
By exploring a common invariant subspace, the original problem is decomposed into a network coupled LQR problem of finite dimension and a decoupled infinite dimensional LQR problem. Based on this decomposition,  centralized optimal solutions with low complexity and nodal collaborative optimal control solutions which employ the projected (or aggregate) information of the states of all agents and the information of the nodal state are established. 

 {The main contributions of this work include the following: First, a solution method for a more general class of graphon LQR problems than in \cite{ShuangPeterCDC19W2} is established. More specifically, 
   the coupling operators are only required by Assumption (A5) to share a common (finite dimensional) invariant subspace and do not need to share the same eigenfunctions as in \cite{ShuangAdityaCDC19,ShuangPeterCDC19W2}}. 
    {It is worth noting that although the graphon LQR problem in this paper is infinite dimensional, the framework and the solution method apply to finite dimensional problems of arbitrary sizes since these are special cases of graphon LQR problems.} 
      {Second, the work in the paper demonstrates the powerful role of Assumption (A5) in enabling  low-complexity and scalable solutions for linear quadratic regulation problems.} Finally, a new approximate control is introduced to generate control solutions to graphon LQR problems with general graphon couplings  {which are not necessarily low-rank}, and it can be implemented directly on networks of finite sizes and  allow for uncertainties in the coupling structures.

The key idea for generating the low complexity solutions is to decouple the original linear quadratic control problems and formulate equivalent problems of low complexity. 
The solution idea was used for linear quadratic mean-field control problems  in \cite{arabneydi2015team,arabneydi2016team} (where couplings are of rank one). This idea is further generalized and applied to the control of graphon and network coupled systems in \cite{ShuangPeterCDC19W2, ShuangAdityaCDC19}. % shuangPhDthesis2018 % 
A related recent paper \cite{arabneydi2019deep} discusses decoupling linear quadratic control problems based on state transformations to formulate equivalent problems and applies it to solve risk-sensitive linear quadratic mean-field control problems. Another closely related work \cite{hamilton2012patterned} studies linear control systems with a shared pattern (or network) structure in state, input and output transformations and the corresponding control synthesis problem.

This paper is organized as follows: Section II introduces graphon control systems, their relations to finite network systems, and graphon LQR problems. Section III discusses the invariant subspace of bounded linear operators. In Section IV, the solution method for graphon LQR problems via invariant subspace decompositions are presented. Section V and Section VI establish the properties of the optimal exact control and the approximate control. Section VII presents the application of the solution method to the regulation of coupled harmonic oscillators on graphs with uncertainties.  {Section VIII discusses the complexity of the solution method and Section IX presents topics for future work.} 
\subsubsection*{Notation}
 $\BR$ and $\BR_+$ denote the set of all real numbers and that of all positive reals respectively. Bold face letters (e.g. $\FA$, $\FB$, $\Fu$) are used to represent graphons, compact operators and functions. Blackboard bold letters (e.g. $\BA$, $\BB$) are used to denote linear operators which are not necessarily compact. We use $\FA^\TRANS$ to denote the adjoint operator of $\FA$. Let $\I$ denote the identity operator for infinite dimensional Hilbert spaces and $I$ denote the identity matrix.   We use $\langle\cdot ,  \cdot \rangle$ and $\|\cdot \|$ to represent respectively inner product and norm.
For any $c\in \BR_+$, let $\ES_c$ denote the set of all bounded symmetric measurable functions $\FA: [0,1]^2 \rightarrow [-c,c]$. In this paper, an element in $\ES_c$ is called a ``graphon''.
Any $\FA \in \ES_c$ can be interpreted as a linear operator from $L^2{[0,1]}$ to $L^2{[0,1]}$ (see e.g. \cite{ShuangPeterTAC18}) and the same notation $\FA$ is also used to represent the associated graphon operator. For a Hilbert space $\HS$, $\mathcal{L}(\HS)$ denotes the set of all bounded linear operators from $\HS$ to $\HS$. Let $\otimes$ denote matrix Kronecker product; more explicitly,  the Kronecker product of $A=[a_{ij}]\in \BR^{n\times n}$ and $B=[b_{ij}] \in \BR^{m\times m}$ is given by   $$A\otimes B = \MATRIX{a_{11} B&\dots &a_{1n} B \\
\vdots & \ddots & \vdots\\
a_{n1}B & \dots&a_{nn}B} \in \BR^{nm\times nm}.$$
Finally, let $\oplus$ denote direct sum.

\section{Graphon LQR Problems}
\subsection{State space and operators}
Consider the space  $$(L^2[0,1])^n \triangleq \underbrace{L^2[0,1]\times\ldots \times L^2[0,1]}_{n} $$ with the inner product defined as follows: for $\Fv, \Fu \in (L^2[0,1])^n$,
\begin{equation}
 	\langle\Fu ,\Fv \rangle \triangleq \int_0^1 \langle \Fv(\alpha), \Fu(\alpha) \rangle_{_{\BR^n}} d\alpha = \sum_{i=1}^n \langle \Fv_i, \Fu_i\rangle_{_{L^2[0,1]}}%.% =  \int_0^1  \Fv(\alpha)^\TRANS \Fu(\alpha) d\alpha,
 \end{equation} 
 where $\Fu_i(\cdot) \in L^2[0,1]$ and $\Fu_{(\cdot)}(\alpha) \in \BR^n$ with $i\in\{1,...,n\}$ and $\alpha \in [0,1]$. %
The corresponding induced norm is given by 
\begin{equation*}
	\begin{aligned}
		 \|\Fv\|_{(L^2[0,1])^n} 
		& = \left(\int_0^1 \|\Fv(\alpha)\|_{\BR^n}^2 d\alpha\right)^\frac12
		= \left(\sum_{i=1}^n \|\Fv_i\|^2_{_{L^2{[0,1]}}} \right)^\frac12.
	\end{aligned}
\end{equation*}

{The space $(L^2[0,1])^n$ with the above inner product is a Hilbert space.}

Consider any $D\in \BR^n$ and $\FA \in \ESC$. 
 For any $\Fv \in \left(L^2[0,1]\right)^n$, the operator $D\FA \in \mathcal{L}\left(\left(L^2[0,1]\right)^n\right)$ is defined by the following linear operation
 \begin{equation}\label{eq:operation-vec}
 \begin{aligned}
     	([D\FA]\Fv)(\alpha) & = D \MATRIX{\int_{[0,1]} \FA(\alpha, \beta) \Fv_1(\beta) d\beta \\
 	 \vdots \\
 	\int_{[0,1]} \FA(\alpha, \beta) \Fv_n(\beta) d\beta } 
 	\\
 	&= D \int_0^1 \FA(\alpha, \beta) \Fv(\beta) d\beta, \quad \forall \alpha \in [0,1].
 \end{aligned}
 \end{equation}
For the identity operator $\BI$, the operation of $D \BI \in \mathcal{L}\left(\left(L^2[0,1]\right)^n\right)$ is defined by 
\begin{equation}\label{eq:identity-operation-vec}
	([D\BI]\Fv)(\alpha) = D \MATRIX{ \Fv_1(\alpha) \\
 	 \vdots \\
 	 \Fv_n(\alpha) } 
 	= D \Fv(\alpha),  \quad \forall \alpha \in [0,1]. 
 \end{equation}
Let 
$L^2([0,T];(L^2[0,1])^n)$ denote the Hilbert space of equivalence classes of strongly measurable (in the B\"ochner sense \cite[p.103]{showalter2013monotone}) mappings $[0,T] \rightarrow (L^2[0,1])^n$ that are integrable with norm
$
\| \Fx \|_{L^2([0,T];(L^2[0,1])^n)} =\ (\int_0^T \|\Fx(\cdot, s)\|^2_{(L^2[0,1])^n} ds)^{{1/2}}.
$
Based on the definitions of the operations in \eqref{eq:operation-vec} and \eqref{eq:identity-operation-vec}, the $k$th $(k\geq 1)$ powers of $D\FA$ and $ D \BI$ are respectively given by 
\begin{equation}
	(D\FA)^k =  D^k \FA^k \quad \text{and} \quad (D\BI)^k =  D^k \BI^k = D^k \BI.
\end{equation}

Let $\BA = [L_{\textup{a}} \BI + D_{\textup{a}} \FA]$ with $L_{\textup{a}},D_{\textup{a}} \in \BR^{n\times n}$ and $\FA \in \ESC$. Clearly, $\BA$ is a bounded linear operator from $\left(L^2[0,1]\right)^n$ to $\left(L^2[0,1]\right)^n$. 
 Following \cite{pazy1983semigroups}, $\BA$ is  the infinitesimal generator of the uniformly (hence strongly)  continuous semigroup 
$S_\BA(t)\triangleq e^{\BA t}=\sum_{k=0}^{\infty} \frac{t^k\mathbf{\BA}^k}{k!}, ~ 0\leq t <\infty.$
Therefore, the initial value problem of the graphon  differential equation
\begin{equation}
	\mathbf{\dot{y}_t}={\BA\mathbf{y}_t},\quad \mathbf{y_0} \in \left(L^2[0,1]\right)^n \label{equ:Noncompact-Operator-Differential-Equation}, \qquad  0\leq t <\infty,
\end{equation}
is well defined and has a  solution given by 
$
	\mathbf{y}_t=e^{{\BA}t}\mathbf{y}_0.
$
\begin{lemma}
If  $n\times n$ dimensional matrices $L_{\textup{a}}$ and $D_{\textup{a}}$ commute, then 
$$e^{\BA t} = e^{L_{\textup{a}} \BI  t } e^{D_{\textup{a}} \FA t} = e^{L_{\textup{a}} t}  e^{D_{\textup{a}} \FA t},\quad \forall  t \in \BR , ~ \forall \FA \in \ESC.$$
\end{lemma}
The proof follows that of the matrix exponential case by replacing the definition of matrix exponentials by semigroups corresponding to  bounded linear operators. 

For details on graphons, graphon operators, graphon spaces and the associated cut metric, readers are referred to \cite{lovasz2012large,ShuangPeterTAC18}. 

\subsection{Linear graphon dynamical systems}%: vector case

\begin{definition}[{Graphon Dynamical Systems}]
  {The graphon dynamical system model is given as follows: 
\begin{equation} \label{eq:dynamics}
	 \begin{aligned}
		&\mathbf{\dot{x}_t}= \BA \Fx_t + \BB \Fu_t,~~ t \in [0,T], % \quad \mathbf{x_0} \in L^2[0,1], 	
		\end{aligned}
\end{equation}
with $\BA = [L_{\textup{a}} \BI + D_{\textup{a}}\FA]$ and $\BB = [L_{\textup{b}} \BI +D_{\textup{b}}\FB]$, where  $L_{\textup{a}}$, $ L_{\textup{b}}$, $D_{\textup{a}}$  and  $D_{\textup{b}}$ are constant matrices of dimension $n\times n$,  $\FA$ and $\FB$ are graphons in $\ESC$,  and $\Fx_t \in \left(L^2[0,1]\right)^n$ and ${\Fu_t} \in \left(L^2[0,1]\right)^n$ are respectively  the system state and the control input at time  $t$. The system in \eqref{eq:dynamics} is denoted by $(\BA;\BB)$}
\end{definition}

 Let $C([0,T];(L^2{[0,1]})^n)$ denote the set of continuous mappings from $[0,T]$ to $(L^2{[0,1]})^n$.
A  solution $\Fx \in L^2([0,T];(L^2{[0,1]})^n)$  is called a {\it mild solution} of  \eqref{eq:dynamics} if 
$	\Fx_t=e^{(t-a){\BA}}{\Fx_a}+\int_a^te^{(t-s)\BA}\BB\Fu_sds$
for all $a\leq t$ in $[0,T]$.

\begin{proposition}The system $(\BA; \BB)$ in
\eqref{eq:dynamics} has a unique  mild solution $\mathbf{x}\in C([0,T];(L^2{[0,1]})^n)$ for any $\mathbf{x_0} \in (L^2{[0,1]})^n$ and any $\mathbf{u}\in L^2([0,T]; (L^2{[0,1]})^n)$.
\end{proposition}
\begin{proof}
Since $\BA$ generates a strongly continuous semigroup and $\BB$ is a bounded linear operator on $(L^2{[0,1]})^n$, we obtain this result following \cite[p.385]{bensoussan2007representation}.
\end{proof}
\subsection{Relation to finite dimensional network systems}\label{subsec:rel-finite-with-graphon}
\begin{definition}[{Network Systems}]
  Consider a network of $N$ agents with the following dynamics
\begin{equation} \label{eq:nodal-dynamics}
	\dot x^i_t= L_{\textup{a}} x_t^i + D_{\textup{a}} z_t^i + L_{\textup{b}} u_t^i + D_{\textup{b}} v_t^i, \quad t \in [0,T]
\end{equation}
where $x_t^i\in \BR^n$ is the state of node $i$, $u_t^i \in \BR^n$ represents the control of node $i$, and $L_{\textup{a}}, L_{\textup{b}}, D_{\textup{a}}, D_{\textup{b}} \in \BR^{n\times n}$ {are constant matrices shared by the agents}; here the network coupling of states and that of controls are given by
$$z^i_t = \frac{1}{N}\sum_{j=1}^N a_{ij}x_t^j\quad \text{and} \quad v_t^i = \frac{1}{N}\sum_{j=1}^N b_{ij}u_t^j,$$
where  {$|a_{ij}|\leq c$ and $|b_{ij}|\leq c$}.
\end{definition}
 Note that problems with $m$ control inputs $(m<n)$ for the nodal dynamics in \eqref{eq:nodal-dynamics} can be considered by filling zeros into columns (with indices between $m$ and $n$) of $L_{\textup{b}}$ and $D_{\textup{b}}$ .

Consider the uniform partition $\{P_1, \ldots,P_N\}$ of $[0,1]$ with $P_1 =[0,\frac{1}{N}]$ and $P_k =(\frac{k-1}{N},\frac{k}{N}]$ for $2\leq k\leq N$.   
Define the step function graphon $\SA$ associated with $A_N \triangleq [a_{ij}]$ as
\begin{equation*}
	\SA(\vartheta,\varphi) = \sum_{i=1}^{N} \sum_{j=1}^{N} \Chi_{_{P_i}}(\vartheta)\Chi_{_{P_j}}(\varphi)a_{ij},  \quad (\vartheta,\varphi) \in [0,1]^2,
\end{equation*}
where $\Chi_{_{P_i}}(\cdot)$ represents the indicator function, that is, $\Chi_{_{P_i}}(\vartheta)=1$ if $\vartheta\in P_i$ and $\Chi_{_{P_i}}(\vartheta)=0$ if $\vartheta\notin P_i$. Similarly, define $\SB$ based on $B_N\triangleq [b_{ij}]$.
Let the piece-wise constant function $\Sxt \in (L^2{[0,1]})^n$ corresponding to $x_t \triangleq [{x_t^1}^\TRANS,...,{x_t^N}^\TRANS ]^\TRANS \in \BR^{nN}$ be given by $\Sxt (\vartheta) =\sum_{i=1}^N {\Chi}_{_{P_i}}(\vartheta) x_t^i$, for all $\vartheta \in [0,1].$ 
Similarly define $\Sut \in (L^2{[0,1])^n}$ that corresponds to $u_t\triangleq [{u_t^1}^\TRANS,...,{u_t^N}^\TRANS ]^\TRANS \in \BR^{nN}$.
\begin{figure}[htb]
    \centering
    \includegraphics[width=8.5cm, trim = {23cm 0 0cm 0}, clip]{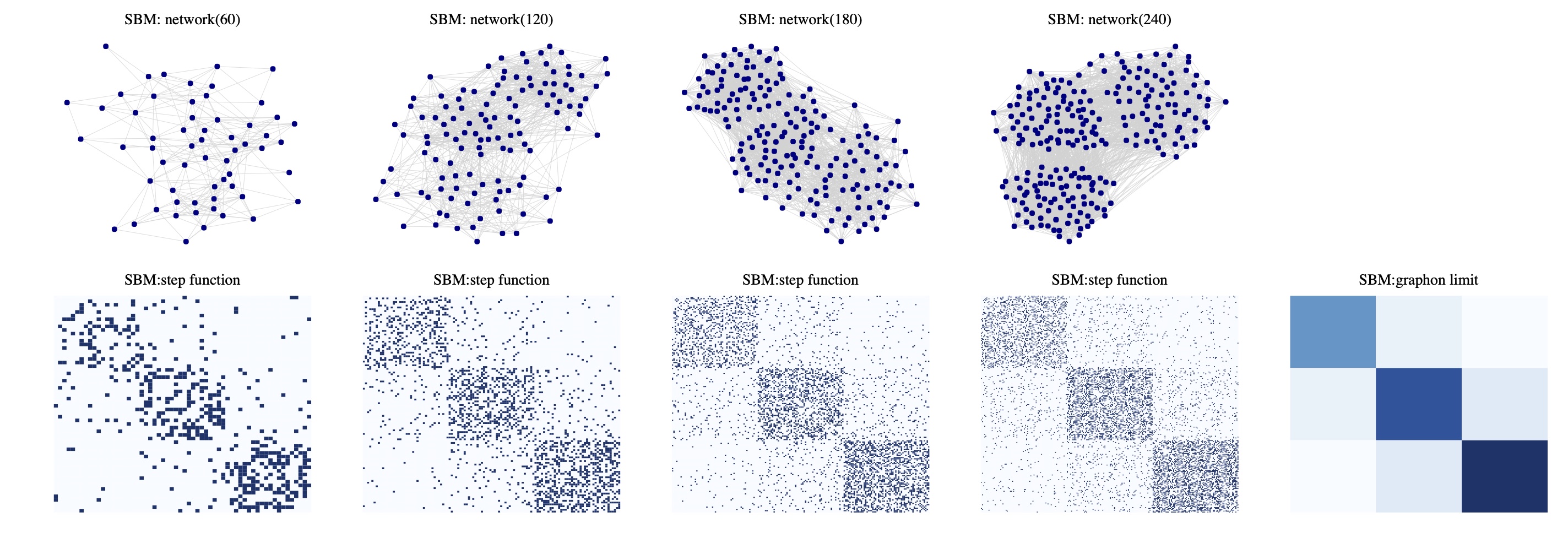}
    \caption{Random graphs generated from a stochastic block model \cite{airoldi2013stochastic}, their {step functions} and the graphon limit. {The underlying stochastic model is characterized by connection probabilities $[0.25, ~ 0.05,~ 0.02;~ 
    0.05,  ~0.35, ~ 0.07;~
    0.02, ~   0.07,~ 0.4]$.} }
    \label{fig:graphs-graphon}
\end{figure} 

Then the corresponding graphon dynamical system for the network system in \eqref{eq:nodal-dynamics} is given by 
\begin{equation}
	\begin{aligned} \label{equ:step-function-dynamical-system}
	&\DSxt= (L_{\textup{a}} \BI+D_{\textup{a}}\SA) \Sxt+(L_{\textup{b}}\BI+D_{\textup{b}}\SB) \Sut,\\ 
	& t\in[0,T],\quad \Sxt, \Sut \in (L^2_{pwc}{{[0,1]}})^n,
	\end{aligned}
\end{equation}
 where   $\SA, \SB \in \ESC$ represent the corresponding graph (i.e. step function graphon) couplings and  $(L^2_{pwc}{{[0,1]}})^n$ represents the set of all piece-wise constant (over each element $P_i$ of the uniform partition) functions  in $(L^2{[0,1]})^n$.

{
 The trajectories of the graphon dynamical system in \eqref{equ:step-function-dynamical-system} correspond one-to-one to the trajectories of the network system in \eqref{eq:nodal-dynamics}. 
{Clearly \eqref{equ:step-function-dynamical-system} is a special case of the graphon dynamical system in \eqref{eq:dynamics}. Therefore any (finite dimensional) network system of arbitrary size in \eqref{eq:nodal-dynamics} can be represented by the graphon dynamical system in \eqref{eq:dynamics}.}
Moreover, the system in \eqref{eq:dynamics} can represent the limit system for a sequence of network systems represented in the form of \eqref{equ:step-function-dynamical-system} when the underlying step function graphon sequences convergence in the operator norm or $L^2{[0,1]^2}$ metric. %~\cite{ShuangPeterTAC18}. {(as illustrated in Fig. \ref{fig:graphs-graphon})} 
}

\subsection{Optimal control problem} \label{sec:ocp}
The control objective is to obtain the control law $\Fu \in L^2([0,T]; (L^2{[0,1]})^n)$ that minimizes the following cost
\begin{equation}\label{eq:cost}
	J(\Fu) = \int_0^T \left(\langle\Fx_t, \BQ \Fx_t \rangle+ \langle\Fu_t,  \Fu_t \rangle \right) dt + \langle \Fx_T, \BQ_{_T} \Fx_T\rangle,
\end{equation}
where  $\BQ, \BQ_T \in \LS{\left((L^2{[0,1]})^n \right)}$,
subject to the system dynamics in \eqref{eq:dynamics}. 
Consider the following assumption:
\begin{description}
	\item[(A1)] The bounded linear operators $\BQ$ and $\BQ_{_T}$ in $\LS{\left((L^2{[0,1]})^n \right)}$ are Hermitian and non-negative, that is, $\BQ^\TRANS = \BQ$,  $\BQ_T^\TRANS = \BQ_T$ and for any $\Fv\in( L^2[0,1])^n$,  $\langle \Fv,\BQ \Fv\rangle\geq 0,\langle \Fv, \BQ_{_T}\Fv\rangle \geq 0$.
\end{description}
The optimal control problem can be solved via dynamic programming which gives rise to the following Riccati equation: 
\begin{equation}\label{equ: noncompact-Riccati-Equation}
	- \dot{\BP}=\BA^\TRANS \BP+\BP \BA- \BP\BB\BB^\TRANS \BP + \BQ, \quad \BP(T)= \BQ_{_T}.
\end{equation}
%
%It can be proved that the optimal control $\Fu^*$ is related to the optimal state $\Fx^*$ by the feedback formula
%
Given the solution $\BP$ to the Riccati equation, the optimal control $\Fu^*\triangleq \{\Fu^*_t, t\in[0,T]\}$ is given by 
\begin{equation} \label{equ: noncompact-feedback-control-law}
	\Fu^*_t=-\BB^\TRANS \BP(t) \Fx^*_t, \quad t\in [0, T]
\end{equation}
and moreover $\Fx^*\triangleq \{\Fx^*_t, t\in[0,T]\}$ is the solution to the closed loop equation 
\begin{equation} \label{equ: Closed-Loop-System-Non-Compact}
	\begin{aligned}
	& \dot{\Fx}_t=\big(\BA- \BB \BB^\TRANS \BP(t)\big) \Fx_t, 
	& t\in [0,T], \Fx_0 \in (L^2{[0,1]})^n.
	\end{aligned}
\end{equation}
See \cite{bensoussan2007representation} for more details and notice that we reverse the time for the Riccati equation in \cite{bensoussan2007representation}.

\begin{proposition}[{{\cite[Part IV]{bensoussan2007representation}}}] 
\label{prop:existence-unique-opt-sol}
	Under the assumption \textup{(A1)}, there exists a unique solution to the Riccati equation  \eqref{equ: noncompact-Riccati-Equation} and furthermore there exists a unique optimal solution pair $(\Fu^*, \Fx^*)$ as given in \eqref{equ: noncompact-feedback-control-law} and \eqref{equ: Closed-Loop-System-Non-Compact}.
\end{proposition}
{In general for an infinite dimensional operator Riccati equation, there are different types of solutions (such as mild, weak, strict and classical solutions \cite{bensoussan2007representation}). In this current formulation it can be verified that these different types of solutions are equivalent, and hence we don't distinguish among them.}

\section{Invariant Subspace}

Consider a Hilbert space $\HS$ and let ${\LS(\HS)}$ denote the set of %(symmetric compact) 
bounded linear operators from $\HS$ to $\HS$. 
\begin{definition}[{Invariant Subspace}]
	An \emph{invariant subspace} of a bounded linear operator $\mathbb{T}\in \LS(\HS)$  is defined as any subspace $\SBS \subset \HS$ such that 
\[
\forall \Fx\in \SBS, \quad \mathbb{T} \Fx \in \SBS.
\] 
Then the subspace $\SBS$ is said to be \emph{$\mathbb{T}$-invariant}. 
\end{definition}
 By definition any subspace $\SBS \subset \HS$ is $\BI$-invariant. 

Consider a self-adjoint compact linear operator $\FA = \FA^\TRANS \in \LS(\HS)$. An application of the spectral theorem \cite[Chapter 8, Theorem 7.3]{sauvigny2012partial} implies that $\FA$ has non-trivial invariant subspace. 
Any eigenspace of $\FA$ (i.e. the space spanned by some eigenfunctions) is an invariant subspace of $\FA$. 
Let the Hilbert space $\HS$ be decomposed by $\SBS_{iv}$ and its orthogonal complement $(\SBS_{iv})^\perp$  as follows 
	$\HS = \SBS_{iv} \oplus (\SBS_{iv})^\perp,$
where $\SBS_{iv}$ is an invariant subspace of $\FA$.
By the orthogonal decomposition theorem, for any $\Fx\in \HS$, there exists a unique decomposition $\Fx = \Fx^\Ff + \breve \Fx$ with $\Fx^\Ff \in \SBS_{iv}$ and  $\breve \Fx \in (\SBS_{iv})^\perp$.  
We note that for any $ \Fz \in (\SBS_{iv})^\perp$,  $\FA  \Fz \in (\SBS_{iv})^\perp$ holds, since 
for any $\Fu \in \SBS_{iv}$ the following hold:
\begin{equation}
  \langle \Fu, \FA \Fz \rangle = \langle \FA^\TRANS \Fu, \Fz \rangle= \langle\FA \Fu, \Fz\rangle = 0.
\end{equation}
This means that $(\SBS_{iv})^\perp$ is also an invariant subspace of $\FA$.
Therefore the following property holds: 
\begin{equation}
\begin{aligned}
%\FA \Fx &= \FA \Fx^\Ff + \FA \breve \Fx\\
	\langle \Fx, \FA \Fx\rangle &=  \big\langle (\Fx^\Ff+\breve \Fx), \FA (\Fx^\Ff+\breve \Fx)\big\rangle
	 = \big\langle \Fx^\Ff, \FA \Fx^\Ff\big\rangle+ \big\langle \breve \Fx, \FA \breve \Fx\big\rangle. 
\end{aligned}
\end{equation}
The above property holds trivially for the identity operator $\BI$.

Let $\mathcal{S} \subset L^2[0,1]$ be an invariant subspace of $\FA \in \ESC$ and consider the subspace of $(L^2[0,1])^n$ denoted by 
$$(\mathcal{S})^n \triangleq \underbrace{\mathcal{S}\times\ldots \times \mathcal{S}}_{n} \subset (L^2[0,1])^n.$$ 
Clearly, by definition, $(\SBS \oplus \SBS^\perp)^n = (L^2[0,1])^n$. Any $\Fv \in (L^2[0,1])^n$ can be uniquely decomposed {through} its components as
\begin{equation}
    \Fv_i = \Fv^\Ff_i + \breve \Fv_{i}, \quad \forall i \in \{1,..., n\}
\end{equation}
where $\Fv^\Ff_i \in \SBS \subset L^2[0,1]$ and $\breve \Fv_{i} \in \SBS^\perp  \subset L^2[0,1]$. We call this \emph{component-wise decomposition} of $\Fv$ into $(\SBS)^n$ and  $(\SBS^\perp)^n$ and denote it by  $\Fv = \Fv^\Ff + \breve \Fv$ where $\Fv^\Ff \in (\SBS)^n$ and $\breve \Fv \in (\SBS^\perp)^n$. 

\begin{proposition}\label{prop:inv-sub-seperation}
Let $\mathcal{S} \subset L^2[0,1]$ be an invariant subspace of $\FA \in \ESC$.
Then both $(\mathcal{S})^n$ and $(\mathcal{S^\perp})^n$ are $[L_{\textup{a}} \BI+D_{\textup{a}} \FA]$-invariant, that is, 
\begin{equation}\label{eq:inv-sub-operatorLD}
\begin{aligned}
	&\forall \Fv \in (\mathcal{S})^n, \quad~~ [L_{\textup{a}} \BI+D_{\textup{a}} \FA] \Fv  \in  (\mathcal{S})^n;\\
	&\forall \Fv \in (\mathcal{S^\perp})^n, \quad [L_{\textup{a}} \BI+D_{\textup{a}} \FA] \Fv  \in  (\mathcal{S^\perp})^n.
\end{aligned}
\end{equation}
\end{proposition}
Furthermore, for any $\Fv \in (L^2[0,1])^n$, the following decomposition holds  
\begin{equation}\label{eq:ort-dec-operatorLD}
\begin{aligned}
	\langle \Fv,&[L_{\textup{a}} \BI+D_{\textup{a}} \FA] \Fv\rangle 	 \\
	&= \big\langle \Fv^\Ff, [L_{\textup{a}} \BI+D_{\textup{a}} \FA] \Fv^\Ff\big\rangle+ \big\langle \breve \Fv,[L_{\textup{a}} \BI+D_{\textup{a}} \FA] \breve \Fv\big\rangle,
\end{aligned}
\end{equation}
where $\Fv = \Fv^\Ff + \breve \Fv$, $\Fv^\Ff \in (\mathcal{S})^n$ and $\breve \Fv \in  (\SBS^\perp)^n$.
\begin{proof}
\eqref{eq:inv-sub-operatorLD} is obtained by explicitly carrying out the calculations.  
\eqref{eq:inv-sub-operatorLD} together with the fact that $(\mathcal{S})^n$ and $(\mathcal{S^\perp})^n$ are also orthogonal to each other implies \eqref{eq:ort-dec-operatorLD}.
\end{proof}

\section{Solution via Subspace Decomposition}
\subsection{Dynamics and cost decomposition}
Consider the following assumptions:
\begin{description}
	\item [(A2)] $\FA \in \ESC $ and $\FB \in \ESC $ share the same invariant subspace  $\SBS \subset  L^2([0,1])$.% and $\SBS^\perp \subset  L^2([0,1])$. 
\end{description}

\begin{description}
	\item [(A3)] $\BQ = L_{\textup{q}}\BI + D_{\textup{q}}\FQ $ and $\BQ_T = L_{\textup{q}_\textup{T}} \BI + D_{\textup{q}_\textup{T}}\FQ_T$, where $L_{\textup{q}},L_{\textup{q}_\textup{T}},D_{\textup{q}},D_{\textup{q}_\textup{T}} \in \BR^{n\times n}$ and $\FQ, \FQ_{_T} \in \ESC$; 
%\item[(A2)]
$~\FQ$ and $\FQ_T$ share the invariant subspaces $\SBS \subset  L^2([0,1])$.% and $\SBS^\perp \subset  L^2([0,1])$. 
\end{description}

\begin{description}
	\item [(A4)] The invariant subspace $\SBS$ in (A2) and (A3) is the same.
\end{description}

Denote the component-wise decomposition of $\Fx_t \in (L^2[0,1])^n$ as $\Fx_t=\breve \Fx_t + \Fx^\Ff_t$ where $\breve \Fx_t \in (\mathcal{S}^\perp)^n$ and $\Fx_t^\Ff \in (\mathcal{S})^n$. Similarly, define $\Fu_t^\Ff$ and $\breve \Fu_t$.

\begin{lemma}\label{lem:dynamics-decouple}
Under the assumption \textup{(A2)}, the dynamics in \eqref{eq:dynamics} can be decoupled as follows: 
\begin{align}
		\dot \Fx_t^\Ff = & [L_{\textup{a}} \BI +D_{\textup{a}}\FA] \Fx_t^\Ff + [L_{\textup{b}} \BI + D_{\textup{b}}\FB]\Fu_t^\Ff, \label{eq:dynamic-decouple-1}
		\\
		\dot {\breve \Fx}_t = & [L_{\textup{a}} \BI + D_{\textup{a}}\FA] \breve \Fx_t+ [L_{\textup{b}} \BI +D_{\textup{b}}\FB]\breve \Fu_t. \label{eq:dynamic-decouple-2}
\end{align}
\end{lemma}
\begin{proof}
If (A2) holds, then both $\SBS$ and $\SBS^\perp$ are the common invariant subspaces of $\FA$ and $\FB$. 
{Following Proposition \ref{prop:inv-sub-seperation}}, this implies that $(\mathcal{S})^n$ and $({\SBS}^\perp)^n$ are both the common invariant subspaces of $[L_{\textup{a}} \BI +D_{\textup{a}}\FA]$ and $[L_{\textup{b}} \BI + D_{\textup{b}}\FB]$. Furthermore, $(\mathcal{S})^n$ and  $({\SBS}^\perp)^n$ are orthogonal to each other, and  the state $\Fx_t \in  (L^2[0,1])^n$ and the control $\Fu_t \in  (L^2[0,1])^n$ both admit unique component-wise decompositions into $(\mathcal{S})^n$ and  $({\SBS}^\perp)^n$.  These lead to the desired decomposition of the dynamics.
\end{proof}

\begin{lemma}\label{lem:cost-decouple}
Under the assumption \textup{(A3)}, %and (A4), 
the cost in \eqref{eq:cost} can be decoupled as follows:
\begin{equation}\label{eq:cost-sep}
		J(\Fu) = J_\mathcal{S}(\Fu^\Ff) + J_\mathcal{S^\perp}(\breve \Fu), 
	\end{equation}	
	where 
	\begin{equation}
		J_\mathcal{S}(\Fu^\Ff) = \int_0^T\left(\langle\Fx_t^\Ff, \BQ \Fx_t^\Ff \rangle+ \langle\Fu_t^\Ff,  \Fu_t^\Ff \rangle \right) dt + \langle \Fx_T^\Ff, \BQ_T \Fx_T^\Ff\rangle \label{eq:cost-decouple-1}
    \end{equation}
    \begin{equation}
		J_\mathcal{S^\perp}(\breve \Fu)=\int_0^T\left(\langle \breve  \Fx_t, \BQ \breve  \Fx_t \rangle+ \langle \breve \Fu_t,  \breve  \Fu_t \rangle \right) dt + \langle \breve  \Fx_T, \BQ_T \breve  \Fx_T\rangle. \label{eq:cost-decouple-2}
	\end{equation}~
\end{lemma}

\begin{proof}
	Under the assumption (A3),  {an application of Proposition \ref{prop:inv-sub-seperation} yields} 
\[
\begin{aligned}
	 &\langle\Fx_t, \BQ \Fx_t \rangle= \langle\Fx_t^\Ff, \BQ \Fx_t ^\Ff\rangle + \langle\breve \Fx_t, \BQ \breve \Fx_t\rangle, \\
	 &\langle\Fu_t, \Fu_t \rangle = \langle\Fu_t^\Ff,  \Fu_t ^\Ff\rangle + \langle\breve \Fu_t, \breve \Fu_t\rangle,\\
	 &\langle\Fx_t, \BQ_T \Fx_t \rangle = \langle\Fx_t^\Ff, \BQ_T \Fx_t ^\Ff\rangle + \langle\breve \Fx_t, \BQ_T \breve \Fx_t\rangle. \\
\end{aligned}
\]
The above separations hold for any $t\in[0,T]$ and hence we obtain the cost decomposition in \eqref{eq:cost-sep}.
\end{proof}

Under assumptions (A1)-(A4), the original problem can be decoupled as separate LQR problems in orthogonal subspaces, that is,  the LQR problem defined by \eqref{eq:dynamic-decouple-1} and \eqref{eq:cost-decouple-1}, and the LQR problem given by \eqref{eq:dynamic-decouple-2} and \eqref{eq:cost-decouple-2}. These problems can be solved independently {and each of them has a unique solution}. 
\subsection{Low-complexity solutions}
In certain situations, the above decoupling leads to simplifications. 
\begin{description}
	\item [(A5)] (i) The common invariant subspace $\SBS$ in (A4) of the underlying coupling operators $\FA$, $\FB$, $\FQ$,  and $\FQ_T$ is finite-dimensional; \\
	(ii) Furthermore, the underlying coupling operators $\FA$, $\FB$, $\FQ$,   and $\FQ_T$ admit exact low-rank representations in $\SBS$, that is, for any $\breve \Fv \in \SBS^\perp$, $\FA \breve\Fv =0$, $\FB \breve\Fv =0$, $\FQ \breve\Fv =0$ and $\FQ_T\breve\Fv =0$.
\end{description}
{The smallest subspace $\mathcal{S}$ that satisfies (A5) is defined as  the \emph{nontrivial common invariant subspace} of operators $\FA, \FB, \FQ$ and $\FQ_T$.}

{Assumption (A5) is satisfied in many cases. For instance, in control or game problems with mean field couplings, the common invariant subspace is naturally $\mathcal{S}=\textup{span}\{\mathbf{1}\}$ where $\mathbf{1}\in L^2[0,1]$ with $\mathbf{1}(\alpha) = 1$ for all $\alpha \in[0,1]$, and hence (A5) is satisfied. A second example would be that where the underlying couplings may be multi-hop neighbourhood couplings on a single network as illustrated in  \cite{ShuangAdityaCDC19,ShuangPeterCDC19W2} and the underlying common invariant subspace is just the eigenspace of the coupling matrix. As another example,  coupling matrices (or similarity matrices) in recommender systems \cite{aggarwal2016recommender} may be built upon certain low-dimensional feature space (or latent factors), which may naturally give rise to a common invariant subspace.} 

{
For conditions and algorithms to identify common invariant subspaces of matrices and linear operators readers are referred to \cite{arapura2004common,drnovvsek2001common,jamiolkowski2015generalized} and the references therein. 
}

 The result below follows Lemma \ref{lem:dynamics-decouple}.
\begin{corollary}
	Under assumptions \textup{(A2)} and \textup{(A5)}, the dynamics in \eqref{eq:dynamics} can be decoupled as follows: 
\begin{align}
		\dot \Fx_t^\Ff = & [L_{\textup{a}} \BI +D_{\textup{a}}\FA] \Fx_t^\Ff + [L_{\textup{b}} \BI +D_{\textup{b}}\FB]\Fu_t^\Ff, \label{eq:decouple-dyn-lowdim}
		\\
		\dot {\breve \Fx}_t = & [L_{\textup{a}} \BI]  \breve \Fx_t+ [L_{\textup{b}} \BI]  \breve \Fu_t. \label{eq:aux-dya-lowdim}
\end{align}
\end{corollary}

%$$
An application of Lemma \ref{lem:cost-decouple} yields the following result. 
\begin{corollary}
Under the assumptions \textup{(A3)} and \textup{(A5)}, the cost in \eqref{eq:cost} can be decoupled as follows:
\begin{equation}
		J(\Fu) = J_\mathcal{S}(\Fu^\Ff) + J_\mathcal{S^\perp}(\breve \Fu), 
	\end{equation}	
	where 
	\begin{align}
		J_\mathcal{S}(\Fu^\Ff) &= \int_0^T\left(\langle\Fx_t^\Ff, \BQ \Fx_t^\Ff \rangle+ \langle\Fu_t^\Ff,  \Fu_t^\Ff \rangle \right) dt + \langle \Fx_T^\Ff, \BQ_T \Fx_T^\Ff\rangle \label{eq:decouple-cost-lowdim}, \\
		J_\mathcal{S^\perp}(\breve \Fu)&=\int_0^T\left(\langle \breve  \Fx_t, [L_{\textup{q}}\BI] \breve  \Fx_t \rangle+ \langle \breve \Fu_t, \breve  \Fu_t \rangle \right) dt + \langle \breve  \Fx_T, [L_{\textup{q}_\textup{T}}\BI] \breve  \Fx_T\rangle. \label{eq:aux-cost-lowdim}
	\end{align}
\end{corollary}

\subsection{Projection into a low-dimensional subspace}
Consider an arbitrary orthonormal basis $\{\Ff_1\ldots\Ff_d\}$  for the low-dimensional subspace $\SBS_\Ff \subset L^2[0,1]$ of dimension $d$. Note that $\Ff_1\ldots\Ff_d$ are not necessarily the eigenfunctions of the operator $\FA$, $\FB$, $\FQ$, or $\FQ_T$.
For all $\ell, k\in \{1,\ldots,d\}$, let
$$
\begin{aligned}
A_{\ell k}&\triangleq \langle \FA \Ff_\ell, \Ff_k\rangle,~~ B_{\ell k}\triangleq \langle \FB \Ff_\ell,\Ff_k\rangle, \\%~~ R_{ij}\triangleq \langle \FB \Ff_i, \Ff_j\rangle\\
Q_{\ell k}&\triangleq \langle \FQ \Ff_\ell, \Ff_k\rangle, ~~Q_{T_{\ell k}}\triangleq \langle \FQ_T \Ff_\ell, \Ff_k\rangle.\\
% x_t^{p\ell} &\triangleq \langle \Fx_t, \Ff_\ell \rangle, ~~ u_{it}^{p\ell} \triangleq \langle \Fu_t, \Ff_\ell \rangle.
\end{aligned}
$$
Denote  $\Ff \triangleq \{\Ff_1,\ldots, \Ff_d\}$.
Consider the following projections 
\begin{align}
	& \text{Proj}_\Ff(\cdot): (L^2[0,1])^n \rightarrow \BR^{nd}, \nonumber\\
	&	\text{Proj}_\Ff(\cdot): {\LS}((L^2[0,1])^n)\rightarrow \BR^{nd\times nd}\nonumber
\end{align} 
 into the subspace $(\SBS_\Ff)^n$ with $\SBS_\Ff= \textup{span}\{\Ff_1,\ldots,\Ff_d\}$. We use the same symbol $\text{Proj}_\Ff(\cdot)$ for the projections of functions and operators as it will be clear  which projection is used in the specific context.
 The projection operations are defined as follows: for  $\Fx_t \in (L^2[0,1])^n$ and any $ D \mathbb{T} \in \LS((L^2[0,1])^n)$ with $ \mathbb{T} \in \LS(L^2[0,1])$ and $D\in \BR^n$:
\begin{equation}
	\begin{aligned}
	&\text{Proj}_\Ff(D\mathbb{T} ) \triangleq  \MATRIX{\langle\Ff_1 ,\mathbb{T}\Ff_1 \rangle& \ldots &\langle\Ff_d ,\mathbb{T}\Ff_d \rangle\\
\vdots &\ddots  &\vdots\\
\langle\Ff_d ,\mathbb{T}\Ff_1 \rangle& \ldots & \langle\Ff_d ,\mathbb{T}\Ff_d \rangle} \otimes D \in \BR^{nd\times nd}, \\[0.3em]
&\text{Proj}_\Ff(\Fx_t) \triangleq x_{t}^{p} = [{x_{1t}^{p}}^\TRANS,\ldots, {x_{nt}^{p}}^\TRANS]^\TRANS \in \BR^{nd},
	\end{aligned}
\end{equation}
where $	x_{it}^{p}  \triangleq [x_{it}^{p1},\ldots, x_{it}^{pd}]^\TRANS, ~ x_{it}^{p\ell}  \triangleq \langle \Fx_{it},\Ff_\ell \rangle,$   $i \in \{1,...,n\},$ 
and $ \Fx_{it} \in L^2[0,1]$ represents the $i$th function component of $\Fx_{t} \in (L^2[0,1])^n$.
According to this definition, we obtain
\begin{equation}
    \begin{aligned}
        &\text{Proj}_\Ff(D\BI)  = I\otimes D ,\quad &\text{Proj}_\Ff(D\FA) = A\otimes D,  \\
        &\text{Proj}_\Ff(D\FB) = B\otimes D,\quad &\text{Proj}_\Ff(D\FQ) = Q\otimes D,\\
        &\text{Proj}_\Ff(D \FQ_T) = Q_T \otimes D,& \\
    \end{aligned}
\end{equation}
for any $D\in \BR^{n\times n}$.%, for which the following lemma holds.
\begin{lemma}
	If $\SBS_\Ff \triangleq \textup{span}\{\Ff_1,\ldots,\Ff_d\}$ forms an invariant subspace of $\FA\in \ESC$, then for any $D\in \BR^{n\times n}$ and $\Fx_t \in (L^2[0,1])^n$, the following relations  hold
\begin{equation*}
\begin{aligned}
	\textup{Proj}_\Ff(D\FA \Fx_t)&= \textup{Proj}_\Ff(D\FA )\textup{Proj}_\Ff(\Fx_t) = (A \otimes D) x_t^p \in \BR^{nd};%\\
\end{aligned}
\end{equation*}	
{Moreover, if for any $\Fv\in \mathcal{S}_\Ff^\perp $, $\FA\Fv =0$, then}	
\begin{equation*}
\begin{aligned}
   \langle [D\FA] \Fx_t, \Fx_t \rangle  &=  \textup{Proj}_\Ff(\Fx_t)^\TRANS  \textup{Proj}_\Ff(D\FA )\textup{Proj}_\Ff(\Fx_t) 
    \\&= (x_t^p)^\TRANS (A\otimes D) x_t^p.
\end{aligned}
\end{equation*}
\end{lemma}

For any $v\in \BR^n$ and $\Fz \in L^2[0,1]$, let $v\Fz \in (L^2[0,1])^n$ be defined as follows: for any $\alpha \in [0,1], i \in \{1,\ldots,n\}$,
\begin{equation}
	(v\Fz)(\alpha) = v \Fz(\alpha),  \quad   (v\Fz)_i = v_i \Fz.
\end{equation}
Let the $i$th component of $\Fx_{t}^\ell \in (L^2[0,1])^n$ be defined by 
	$\Fx_{it}^\ell =  \langle \Fx_{it},\Ff_\ell \rangle \Ff_\ell = x_{it}^{p\ell} \Ff_\ell$.
\begin{proposition}\label{eq:decoupled-lqr-problems}
Under \textup{(A1)-(A5)}, the original problem defined by \eqref{eq:dynamics} and \eqref{eq:cost} can be transformed into the following equivalent problem
\begin{align}
		\dot x_t^{p} = & (I\otimes L_{\textup{a}} + A\otimes D_{\textup{a}}) x_t^{p} + (I\otimes L_{\textup{b}}  +B\otimes D_{\textup{b}})u_t^{p},  \label{eq:dynamics-xp}
		\\
		\dot {\breve \Fx}_t^\gamma = & L_{\textup{a}}  \breve \Fx_t^\gamma+ L_{\textup{b}} \breve \Fu_t^\gamma,  ~~\gamma \in  [0,1], \label{eq:dynamics-xp-aux}
\end{align}
 with the following cost to be minimized
	\begin{align}
			J_\mathcal{S}(u^{p}) &= \int_0^T\left(x_t^{p\TRANS} {(I\otimes L_{\textup{q}}+Q\otimes D_{\textup{q}} )} x_t^p + u_t^{p\TRANS} u_t^p \right) dt \nonumber\\
		& \qquad + x_T^{p\TRANS} (I\otimes L_{\textup{q}_\textup{T}} +{Q}_T\otimes D_{\textup{q}_\textup{T}}) x_T^p,  \label{eq:decouple-cost-xp}\\
		J_\mathcal{S^\perp}({\breve\Fu}^\gamma)&=\int_0^T\left(  \breve\Fx_t^{\gamma\TRANS} L_{\textup{q}} \breve\Fx_t^\gamma + {\breve \Fu}_t^{\gamma\TRANS} \breve \Fu_t^\gamma\right) dt
		 +  {\breve\Fx}_T^{\gamma\TRANS} L_{\textup{q}_\textup{T}} {\breve \Fx_T}^\gamma,\label{eq:cost-xp-aux}
	\end{align}
	 \text{for almost all} $\gamma \in [0,1]$ where  $\breve \Fx_t^\gamma, \breve\Fu_t^\gamma \in \BR^{n}, x_t^p,u_t^p \in\BR^{nd}$, and the initial conditions are given by $x_0^p = \textup{Proj}_\Ff(\Fx_0)={[{x_{1 0}^{p}}^\TRANS,\ldots, {x_{n0}^{p}}^\TRANS]}^\TRANS$ and $\breve \Fx_0^\gamma = \Fx_0^\gamma -\sum_{\ell =1}^d x_{\ell0}^p \Ff_\ell(\gamma).$ %  and $$ with  $\Fx_{i0}^\ell =   x_{i0}^{p\ell} \Ff_\ell$.
\end{proposition}
\begin{proof}
By performing $\text{Proj}_{\Ff}(\cdot)$ on both sides of \eqref{eq:decouple-dyn-lowdim}, we obtain
\eqref{eq:dynamics-xp}.
 The same projection of \eqref{eq:decouple-cost-lowdim} results in \eqref{eq:decouple-cost-xp}.  
The auxiliary problem defined by \eqref{eq:aux-dya-lowdim} and \eqref{eq:aux-cost-lowdim} is the same as the problem defined by \eqref{eq:dynamics-xp-aux} and \eqref{eq:cost-xp-aux} \text{for almost all} $\gamma \in [0,1]$, since the definition of the dynamics is pointwise  and 
%\begin{equation*}
	$J_\mathcal{S^\perp}(\breve \Fu) = \int_0^1 J_\mathcal{S^\perp}(\breve \Fu^\gamma) d\gamma$
%\end{equation*}
where $J_\mathcal{S^\perp}(\breve \Fu^\gamma)\geq 0$ for each $\gamma \in  [0,1]$, that is, $J_\mathcal{S^\perp}(\breve \Fu)$ is the convex combination of all (non-negative) elements in $\{J_\mathcal{S^\perp}(\breve \Fu^\gamma): \gamma \in [0,1]\}$.
\end{proof}

\begin{figure*}[htb] 
\centering
\subfloat[Projections and auxiliary signals of both state and control]{
	\includegraphics[height=5.5cm,trim = {1cm 0 7.7cm 0}, clip]{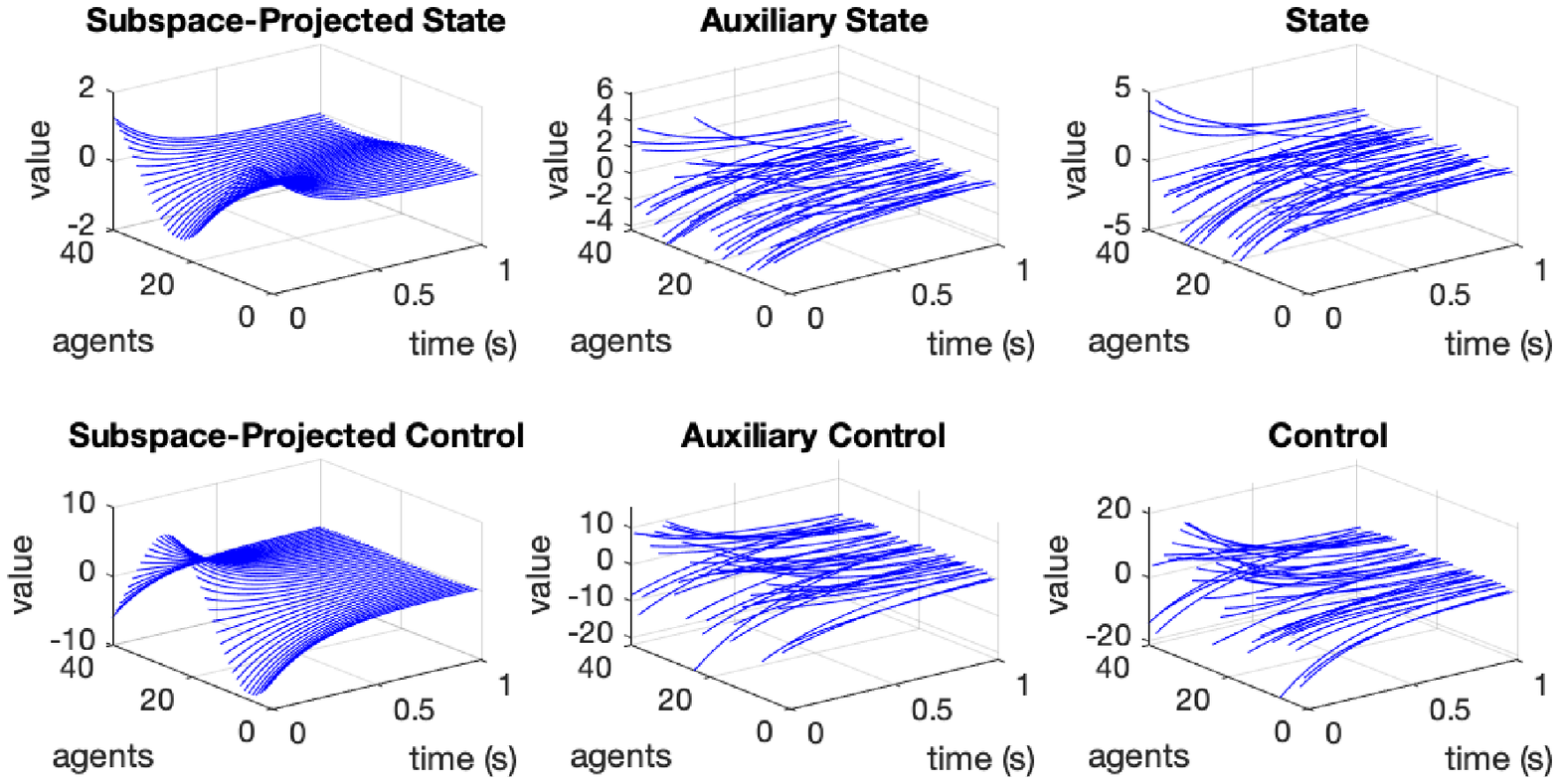}}	\quad
\subfloat[Comparison between the exact control and the centralized optimal control]{
	\includegraphics[height=5.5cm,trim = {1cm 0 1.5cm 0}, clip]{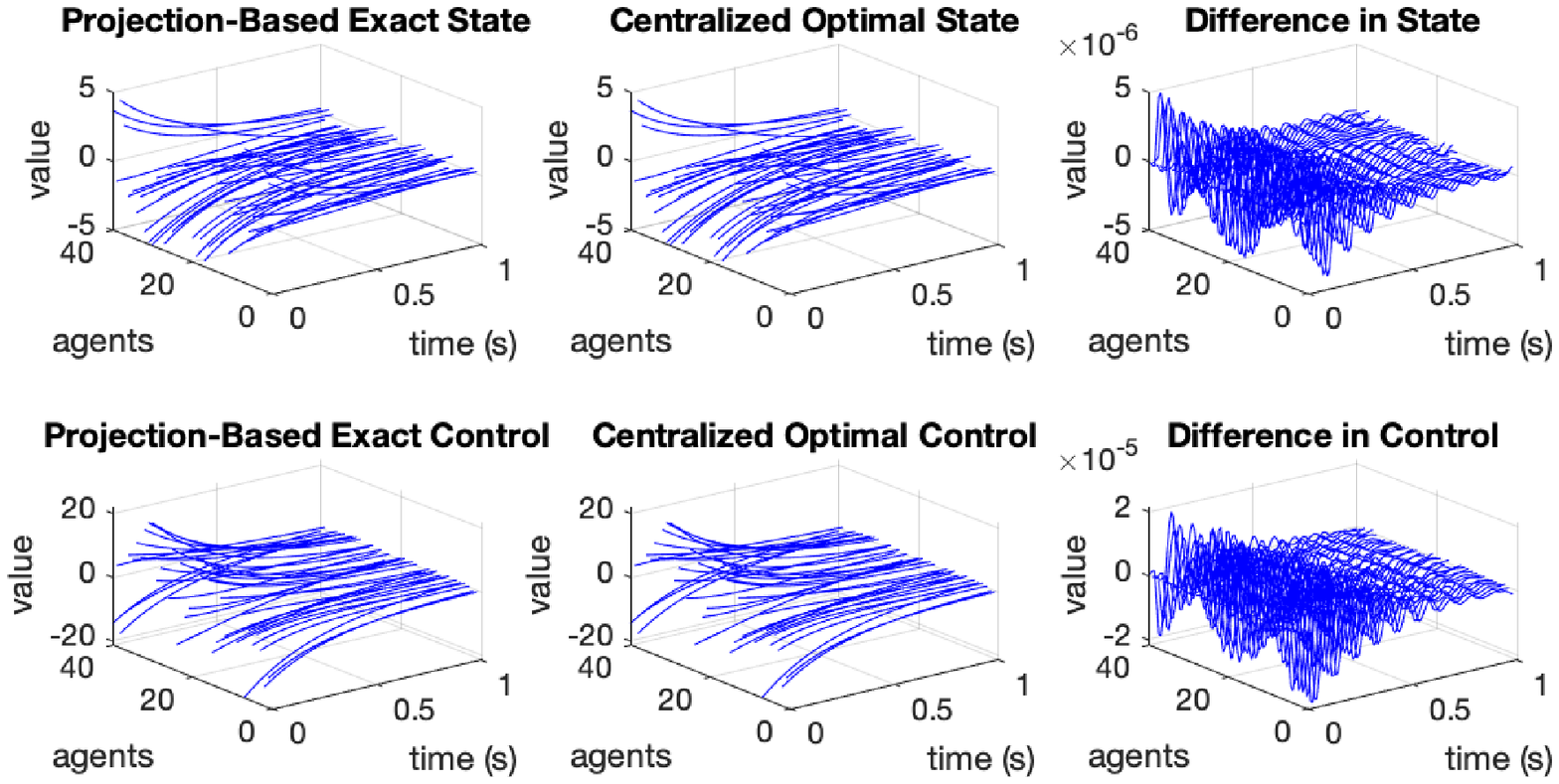} } 

	\caption{%This is a simulation based on the exact control in Section \ref{sec:exact-control}. 	
	{The term ``projection-based" refers to the result generated from the method developed in Section \ref{sec:exact-control}, while centralized optimal controls are generated by directly solving the problem with state space discretization. }
	 {The differences between the trajectories generated based on these two different solution methods are  mainly due to numerical approximations and are shown in the last row of (b). Following the projection-based solution, the trajectories for the projected dynamics and those for  the auxiliary dynamics are respectively shown in the 1st and 2nd columns of (a).}}  
	\label{fig:simulation}
\end{figure*}

At this stage, the optimal control of an infinite network can be solved based on the optimal control solution to two decoupled LQR problems, where one requires solving a Riccati equation of dimension ${nd\times nd}$ and the other requires solving a Riccati equation of dimension $n\times n$.  We note that the system dynamics \eqref{eq:dynamics-xp-aux} of the auxiliary problem is infinite dimensional as $\gamma$ takes values in the interval $[0,1]$.

\section{Exact Control}\label{sec:exact-control}
\begin{theorem}\label{thm:main-result}
	Under \textup{(A1)-(A5)},  the optimal control for the original problem defined by \eqref{eq:dynamics} and \eqref{eq:cost} is unique and is given by 
\begin{equation}\label{eq:control-combined}
	\Fu_t^o  = \breve \Fu_t^o + \sum_{\ell=1}^du_t^{op\ell}\Ff_\ell =:  \breve \Fu_t^o + \Ff \circ u_t^{op}
\end{equation}
where $\Ff =\{\Ff_1,\ldots, \Ff_d\}$,
\begin{equation} \label{eq:opt-control-sub}
\begin{aligned}
u_t^{op} =& -(I\otimes L_{\textup{b}}  +B\otimes D_{\textup{b}})^\TRANS \Pi_t x_t^{op},\\
	-\dot\Pi_t =& (I\otimes L_{\textup{a}} + A\otimes D_{\textup{a}}) ^\TRANS \Pi_t + \Pi_t (I\otimes L_{\textup{a}} + A\otimes D_{\textup{a}}) \\
	&   - \Pi_t (I\otimes L_{\textup{b}}  +B\otimes D_{\textup{b}})(I\otimes L_{\textup{b}}  +B\otimes D_{\textup{b}})^\TRANS \Pi_t\\
	& + (I\otimes L_{\textup{q}} +{Q}\otimes D_{q}), \\
	 \Pi_T =& I\otimes L_{\textup{q}_\textup{T}} +{Q}_T\otimes D_{\textup{q}_\textup{T}},
\end{aligned}
\end{equation}
and the optimal control in the auxiliary direction is given by 
\begin{equation}\label{eq:opt-control-aux}
	\begin{aligned}
		\breve \Fu_t^{o\gamma} &= - L_{\textup{b}}^\TRANS  \pi_t \breve \Fx_t^{\gamma}, \quad \text{for almost all } \gamma \in  [0,1],\\
		-\dot \pi_t &=  L_{\textup{a}}^\TRANS \pi_t + \pi_t L_{\textup{a}} - \pi_t L_{\textup{b}} L_{\textup{b}}^\TRANS\pi_t + L_{\textup{q}}, \quad  
		\pi_T  = L_{\textup{q}_\textup{T}}.
	\end{aligned}
\end{equation}
\end{theorem}
\begin{proof}
We have two decoupled LQR problems: (i) the LQR problem defined by \eqref{eq:dynamics-xp} and \eqref{eq:decouple-cost-xp}, and (ii) the LQR problem given by \eqref{eq:dynamics-xp-aux} and \eqref{eq:cost-xp-aux}. 
	By classical finite dimensional LQR \cite{liberzon2011calculus}, the optimal control law for each LQR problem is unique and is given by \eqref{eq:opt-control-sub} and  \eqref{eq:opt-control-aux}.
	Then recovering the unique decomposition of the control input into the component spaces $(\mathcal{S}^\perp)^n$ and $(\mathcal{S})^n$ as $\Fu_t = \breve \Fu_t  + \Fu_t^{\Ff }$, we obtain the optimal control law in \eqref{eq:control-combined} for the original problem. 
\end{proof}

{Each individual agent may compute its part of the optimal control solution locally and then solve the original optimal control problem collaboratively.}
\begin{corollary}\label{cor:local collaborative}
	Under \textup{(A1)-(A5)},  the nodal collaborative optimal control solution to the original problem defined by \eqref{eq:dynamics} and \eqref{eq:cost} is given as follows: {\text{for almost all} $\gamma \in  [0,1]$,}
\begin{equation}
	\Fu_t^{o\gamma}  = \breve \Fu_t^{o\gamma} + \sum_{\ell=1}^du_t^{op\ell}\Ff_\ell(\gamma) =:  \breve \Fu_t^{o\gamma} + \Ff(\gamma) \circ u_t^{op}
\end{equation}
where $\Ff(\gamma) =[\Ff_1(\gamma),\ldots, \Ff_d(\gamma)]$, $u_t^{op}$ is given by \eqref{eq:opt-control-sub} and $\breve \Fu_t^{o\gamma}$ is given by \eqref{eq:opt-control-aux}.
 \end{corollary}

To implement the collaborative nodal optimal control, each agent needs to know the projection of states into the subspace, i.e. $x_t^p = \text{Proj}_\Ff(\Fx_t^\Ff)$. The projection represents certain aggregate information of the state in certain invariant subspace of the underlying graphon couplings. 
Agent $\gamma \in  [0,1]$ can then compute ${\breve \Fx_t}^\gamma = \Fx_t^\gamma - \Ff(\gamma)\circ x_t^{p}$ together with the local state $\Fx_t^\gamma$; {the state information $x_t^p, t\in[0,T]$, may be  precomputed based on the aggregate initial condition $x_0^p$.} 

For decentralized solutions in a competitive environment, readers are referred to the work on graphon mean field games~\cite{PeterMinyiCDC19GMFG}.
\subsection{Illustrative Example}\label{sec:numerical-example}
Let $\FA$, $\FB$, $\FQ$ and $\FQ_T$ be given by the following: for all  $(x,y) \in [0,1]^2$, %(excepts sets of points with measure zero),
\begin{equation}
\begin{aligned}
&	\FA (x,y) = 2\cos(2\pi(x-y))+\sin(2\pi(x+y)), \\
&  \FB= \cos(2\pi (x+y)),\quad \FQ (x, y) = \sin(2\pi x) \sin(2\pi y),\\& \FQ_T(x,y)= \cos(2\pi x) \cos(2\pi y).
\end{aligned}
\end{equation}
Consider a subspace $\mathcal{S}_\Ff=\text{span}\{\Ff_1, \Ff_2\}$ with $\Ff_1=\sqrt{2}\sin(2\pi \cdot)$ and $\Ff_2=\sqrt{2}\cos(2\pi \cdot)$ in $L^2[0,1]$. Then $\mathcal{S}_\Ff$ is an invariant subspace of $\FA$, $\FB$, $\FQ$ and $\FQ_T$.
Then projecting these operators into the subspace yields
\begin{equation*}
    A = \text{Proj}_\Ff(\FA) = \MATRIX{1 &\frac12\\\frac12&1\\ },\quad B=\text{Proj}_\Ff(\FB)= \MATRIX{-\frac12&0\\0&\frac12\\ },
\end{equation*}
\begin{equation*}
    Q= \text{Proj}_\Ff(\FQ) = \MATRIX{\frac12&0\\0&0\\ },\quad Q_T= \text{Proj}_\Ff(\FQ_T) = \MATRIX{0&0\\0&\frac12\\ }.
\end{equation*}
Obviously, the projections of these coupling operators into $(\SBS_\Ff)^\perp$ is zeros.  Hence (A5) is satisfied. 
Let $n=1$, $L_{\textup{a}}=2$, $L_{\textup{b}} =1.2$, $L_{\textup{q}}=D_{\textup{a}}=D_{\textup{b}}=D_{\textup{q}}=D_{\textup{q}_\textup{T}} =1$ and $L_{\textup{q}_\textup{T}}=2$. 

Following Proposition \ref{eq:decoupled-lqr-problems}, the original LQR problem for the graphon dynamical system with dynamics in \eqref{eq:dynamics} and cost in \eqref{eq:cost} can be transformed into the LQR control problems defined by \eqref{eq:dynamics-xp} \& \eqref{eq:decouple-cost-xp}, and \eqref{eq:dynamics-xp-aux} \& \eqref{eq:cost-xp-aux}. Based on Corollary \ref{cor:local collaborative}, the original problem is solved in the low dimensional subspace and each agent generate its control law and implements it locally.

 A simulation result is demonstrated in Fig. \ref{fig:simulation}; {it was carried out for a graphon dynamical system with step function approximation and state space discretization based on the uniform partition of size 40. 
	  Note that the step function system represents a network system consisting of 40 nodal agents where each agent is indexed by an interval of length $1/{40}$ in $[0,1]$. The initial conditions for all agents are uniformly sampled from $[-5,5]$.  Each agent locally generates its control input according to Corollary \ref{cor:local collaborative}, and solves one $2\times 2$ Riccati equation and one scalar Riccati equation. {As a comparison, the direct solution requires solving a Riccati equation of dimension $40\times 40$.} }
\section{Approximate Control} \label{sec:approxmiate-control}
If Assumption (A5)-(ii) is not satisfied, that is, $\FA, \FB, \FQ$ and $\FQ_T$ do not admit exact low-rank representations in  some common invariant subspace, one may approximate these operators in some finite-dimensional subspace where their eigenvalues are significant, since these operators are (compact) Hilbert-Schmidt integral operators and have discrete spectrum with zero as the only accumulation point. % (see \cite{ShuangPeterCDC19W2} for a detailed discussion). 
{
More explicitly, since for a graphon $\FA \in \ES_c$, we have $\|\FA\|_2 < \infty$ and  hence the operator $\FA$ is a compact operator according to \cite[Chapter 2, Proposition 4.7]{conway2013course}.  Therefore it has a countable spectral decomposition
	$\FA(x,y) = \sum_{i=1}^{\infty} \lambda_\ell \Ff_\ell(x) \Ff_\ell(y), ~ (x,y)\in[0,1]^2,$
where the convergence is in the $L^2{[0,1]^2}$ sense, $\{\lambda_1, \lambda_2,....\}$  is the set of eigenvalues (which are not necessarily distinct) with decreasing absolute values, and $\{\Ff_1, \Ff_2,...\}$ represents the set of the corresponding orthonormal eigenfunctions (i.e. $\|\Ff_\ell\|_2=1$, and $\langle \Ff_\ell, \Ff_k\rangle =0$ if $l\neq k$). 
The only accumulation point of the eigenvalues is zero \cite{lovasz2012large}, that is, $\lim_{\ell\rightarrow \infty} \lambda_\ell =0.$} %

For two graphon operators $\FA_\mathcal{S}$ and $\FA$, $\FA_\mathcal{S}$ is called the \emph{equivalent linear operator} of $\FA$ in $\mathcal{S}$ if  for all $\Fv \in \mathcal{S}$, $\textup{Proj}(\FA \Fv) = \textup{Proj}(\FA_\mathcal{S} \Fv)$ and the range of $\FA_\mathcal{S}$ lies in $\mathcal{S}$. Let $\FA = \FA_\mathcal{S} + \FA_\mathcal{S^\perp} \in \ESC$ where $\FA_\mathcal{S}$ (resp. $\FA_\mathcal{S^\perp}$) is the equivalent linear operator of $\FA$ in $\mathcal{S}$ (resp. $\mathcal{S}^\perp$).  Similarly define $\FB_{\mathcal{S}}$, $\FB_{\mathcal{S}^\perp}$,$\FQ_{\mathcal{S}}$, $\FQ_{\mathcal{S}^\perp}$,$\FQ_{T\mathcal{S}}$ and $\FQ_{T\mathcal{S}^\perp}$.

Following Lemma \ref{lem:dynamics-decouple}, the dynamics can be decoupled as %
\begin{align}
		\dot \Fx_t^\Ff = & [L_{\textup{a}} \BI +D_{\textup{a}}\FA_\mathcal{S}] \Fx_t^\Ff + [L_{\textup{b}} \BI +D_{\textup{b}}\FB_\mathcal{S}]\Fu_t^\Ff, 
		\\
		\dot {\breve \Fx}_t = & [L_{\textup{a}} \BI +D_{\textup{a}}\FA_\mathcal{S^\perp}] \breve \Fx_t+ [L_{\textup{b}} \BI +D_{\textup{b}}\FB_\mathcal{S^\perp}]\breve \Fu_t. 
		\label{eq:aux-actual-dyn}
\end{align}
Applying the control law in Theorem \ref{thm:main-result} will ignore the effect of $\FA_{\mathcal{S^\perp}}$, $\FB_{\mathcal{S^\perp}}$, $\FQ_{\mathcal{S^\perp}}$ and $\FQ_{T\mathcal{S^\perp}}$. A special case of this type of approximation is explored and discussed in \cite{ShuangPeterCDC19W2}.
\begin{figure*}[htb] 
\centering
\subfloat[Projections and auxiliary signals of both state and control under the projection-based approximate control]{
	\includegraphics[height=5.5cm,trim = {1cm 0 7.7cm 0}, clip]{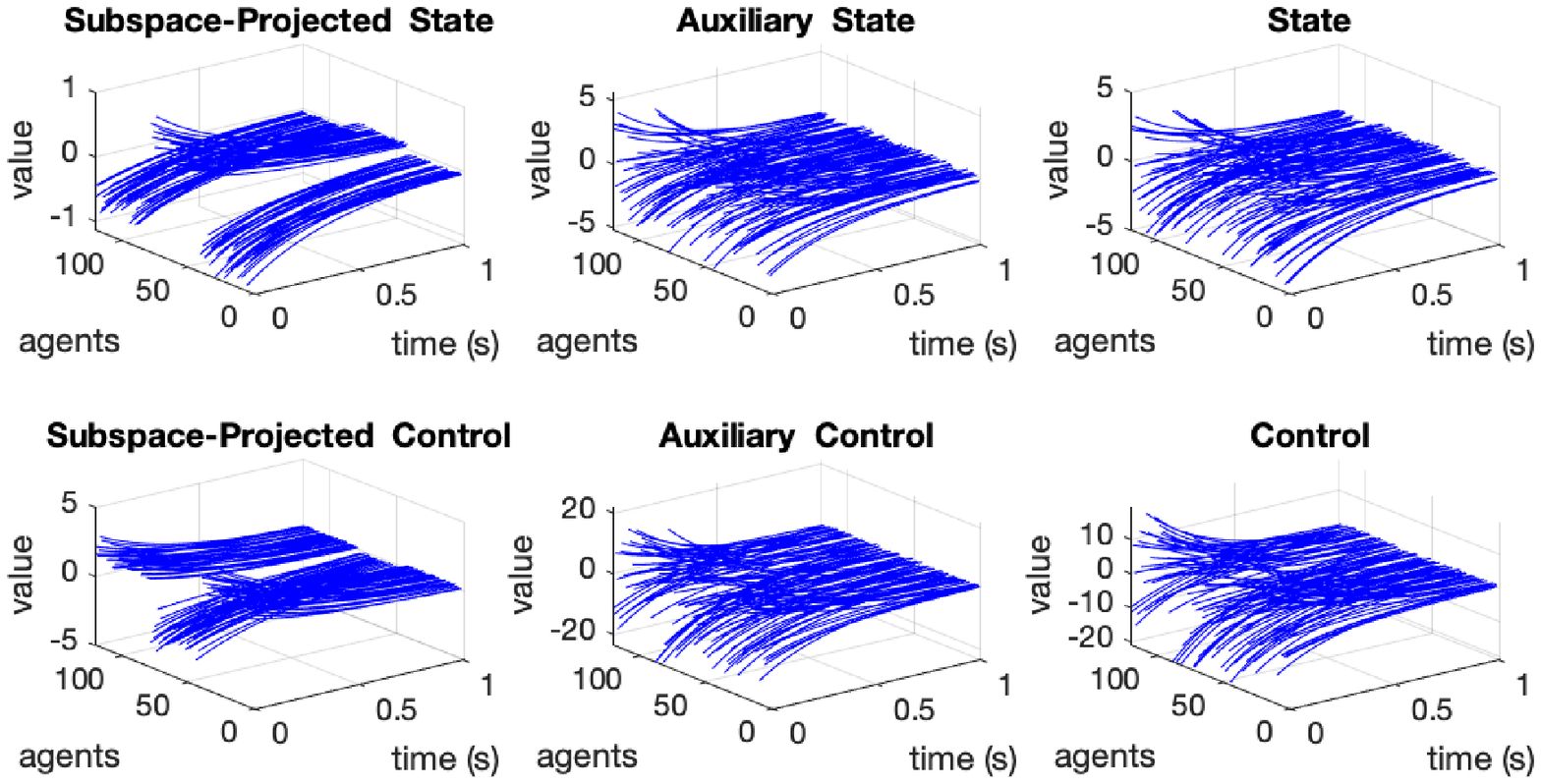}}	
\subfloat[Comparison between the approximate control and the centralized optimal  control]{
	\includegraphics[height=5.5cm,trim = {1cm 0 1.5cm 0}, clip]{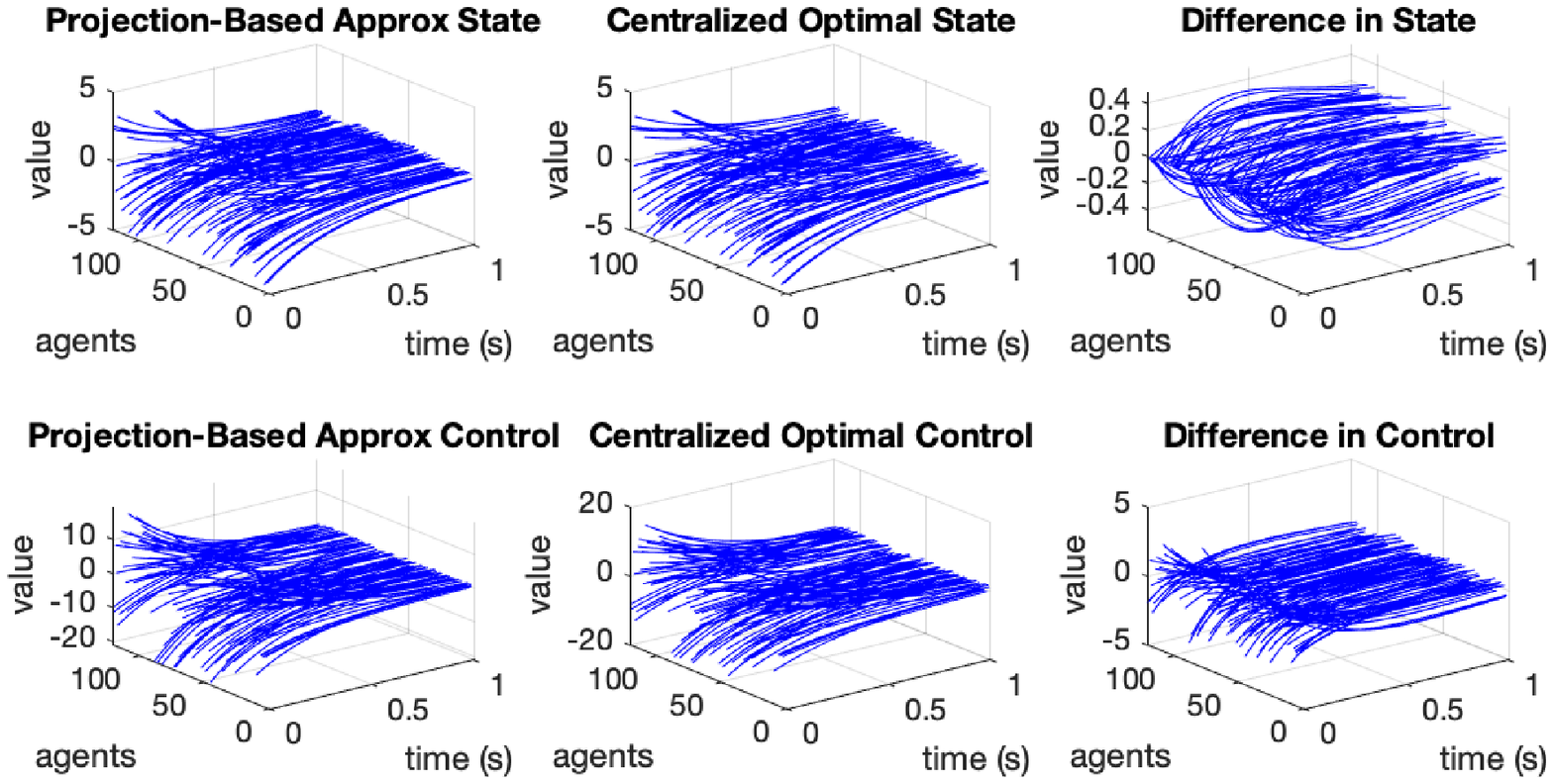} }

	\caption{The term ``projection-based approx" refers to the result generated from the method developed in Section \ref{sec:approxmiate-control}, while ``centralized optimal control or state'' is generated by directly solving the LQR problem. 
    In (a), the first column demonstrates the states and controls associated with projected dynamics under the approximate control, and the second column presents those associated the auxiliary dynamics under  the approximate control.  
 In (b), the first column  illustrates the states and controls under the approximated control, the second column illustrates the states and controls under a direct solution of the LQR problem, and the last column presents the differences in trajectories under two control solutions. The approximate control cost is $2.376\%$ higher than the optimal cost, and the maximum difference in state trajectories is $11.181\%$ of the maximum state under the optimal control. 
    } 
	\label{fig:approximate-control-simulation}
\end{figure*}

To generate approximate control laws that ensure a faster rate of convergence for \eqref{eq:aux-actual-dyn}, a variant of the implementation in Theorem \ref{thm:main-result} can be considered.

\begin{AppImp*}
Consider the case where {$D_{\textup{q}_\textup{T}}>0$, $D_{\textup{q}}\geq 0$, $D_{\textup{b}}L_{\textup{b}}^\TRANS\geq 0$ and the real parts of all the eigenvalues of $D_{\textup{a}}$ are non-negative. Assume the operator norms of $\FA_{\mathcal{S^\perp}}$, $\FB_{\mathcal{S^\perp}}$, $\FQ_{\mathcal{S^\perp}}$ and $\FQ_{T\mathcal{S^\perp}}$ are available for the computation of the control law.}
	Then under \textup{(A1)-(A4)} and \textup{(A5)-(i)},   
	the approximate control law is given by the following
\begin{equation}\label{eq:approx-control-total}
	 \Fu_t = \breve \Fu_t^{app} + \sum_{i=1}^du_t^{opi}\Ff_i =:  \breve \Fu_t^{app} + \Ff \circ u_t^{op}
\end{equation}
where $\Ff =\{\Ff_1,\ldots, \Ff_d\}$,  {$u_t^{op}$ is given by \eqref{eq:opt-control-sub}} and the approximate control $\Fu_t^{app}$ in the auxiliary direction is given by
\begin{equation}\label{eq:approx-opt-control-aux}
	\begin{aligned}
		\breve \Fu_t^{app~\gamma} &= - {L_{\textup{b}}}^\TRANS  \pi_t \breve \Fx_t^\gamma, \quad \text{for almost all } \gamma \in [0,1],\\
		-\dot \pi_t &= (L_{\textup{a}}+ D_{\textup{a}}\|\FA_{\mathcal{S^\perp}}\|_{\textup{op}} )^\TRANS \pi_t + \pi_t (L_{\textup{a}}+ D_{\textup{a}}\|\FA_{\mathcal{S^\perp}}\|_{\textup{op}} ) \\
		&\quad - \pi_t(L_{\textup{b}}^\TRANS L_{\textup{b}}-D_{\textup{b}}L_{\textup{b}}^\TRANS \|\FB_{\mathcal{S^\perp}}\|_{\textup{op}}- L_{\textup{b}}D_{\textup{b}}^\TRANS \|\FB_{\mathcal{S^\perp}}\|_{\textup{op}}) \pi_t \\
		& \quad + L_{\textup{q}}+ D_{\textup{q}} \|\FQ_\mathcal{S^\perp}\|_\textup{op},\\ %+\beta^2 \|\FB\|_{\textup{op}}^2 \pi_t^2, \\
		\pi_T  &= L_{\textup{q}_\textup{T}}+ D_{\textup{q}_\textup{T}} \|\FQ_{T\mathcal{S^\perp}}\|_\textup{op}.
	\end{aligned}
\end{equation}
\end{AppImp*}

The actual dynamics of the auxiliary system is given by \eqref{eq:aux-actual-dyn}. {Since  the operator norms of $\FA_{\mathcal{S^\perp}}$, $\FB_{\mathcal{S^\perp}}$, $\FQ_{\mathcal{S^\perp}}$ and $\FQ_{T\mathcal{S^\perp}}$ are available for  the computation of the control law,} an approximate cost in the auxiliary direction is given by the following form 
\begin{equation}\label{eq:approximate-aux-cost}
    \begin{aligned}
   \tilde{J}_{\mathcal{S^\perp}}(\breve \Fu) 
   &= \int_0^T  \left\{\langle \breve  \Fx_t, [(L_{\textup{q}}  + {D_\textup{q}} \|\FQ_{\mathcal{S^\perp}}\|_{\textup{op}})\BI] \breve  \Fx_t \rangle+ \langle \breve \Fu_t, \breve  \Fu_t \rangle \right\} dt \\& \quad
    +  \langle \breve  \Fx_T, [(L_{\textup{q}_\textup{T}}  + {D_{\textup{q}_\textup{T}}} \|\FQ_{T\mathcal{S^\perp}}\|_{\textup{op}})\BI] \breve \Fx_T\rangle .
    \end{aligned}
\end{equation}
Observe that this cost is always greater than or equal to the actual cost in the auxiliary direction given by
\begin{equation}
\begin{aligned}
   J_{\mathcal{S^\perp}}(\breve \Fu) 
     &= \int_0^T  \left\{\langle \breve  \Fx_t, [L_{\textup{q}} \BI + {D_\textup{q}}\FQ_{\mathcal{S^\perp}}] \breve  \Fx_t \rangle+ \langle \breve \Fu_t, \breve  \Fu_t \rangle \right\} dt \\& \quad
    +  \langle \breve  \Fx_T, [L_{\textup{q}_\textup{T}} \BI + {D_{\textup{q}_\textup{T}}} \FQ_{T \mathcal{S^\perp}}] \breve \Fx_T\rangle.
\end{aligned}
\end{equation}
That is, for all admissible control $\breve \Fu$, $\tilde{J}_{\mathcal{S^\perp}}(\breve \Fu) \geq J_{\mathcal{S^\perp}}(\breve \Fu) $.
The approximate control considered takes the special form
\begin{equation} \label{eq:approx-aux-control}
    \breve \Fu_t^{app}  = - [L_{\textup{b}}^\TRANS \pi_t \BI ]\breve \Fx_t.
\end{equation}
This then yields the closed-loop system dynamics 
\begin{equation}
\begin{aligned}
    \dot{\breve\Fx}_t & = [L_{\textup{a}} \BI +D_{\textup{a}}\FA_\mathcal{S^\perp}] \breve \Fx_t+ [L_{\textup{b}} \BI +D_{\textup{b}}\FB_\mathcal{S^\perp}] [-L_{\textup{b}}^\TRANS \pi_t \BI] \breve\Fx_t.
    \end{aligned}
\end{equation}
Assuming $\pi_{(\cdot)}$ is available (which comes from a Riccati equation to be formulated), by separating the control part, an equivalent closed-loop dynamics is given by
\begin{equation} \label{eq:aux-dynamics-with-special-control}
\begin{aligned}
    \dot{\breve\Fx}_t& =\left[L_{\textup{a}} \BI +D_{\textup{a}}\FA_\mathcal{S^\perp}- D_{\textup{b}} L_{\textup{b}}^\TRANS  \pi_t \FB_\mathcal{S^\perp}\right]\breve \Fx_t+ [L_{\textup{b}}^\TRANS \pi_t \BI]   \breve \Fu_t^{app},
    \end{aligned}
\end{equation}
where $ \breve \Fu_t^{app}  =  - [L_{\textup{b}}^\TRANS \pi_t \BI ]\breve \Fx_t.$
The control solution in \eqref{eq:approx-opt-control-aux} solves optimally the LQR problem with dynamics
\begin{equation} \label{eq:scalar-aux-dynamics}
    \begin{aligned}
    \dot{\breve\Fx}_t& =\left[L_{\textup{a}} \BI +\|\FA_\mathcal{S^\perp}\|_{\textup{op}} D_{\textup{a}}\BI +D_{\textup{b}}L_{\textup{b}}^\TRANS \pi_t \|\FB_\mathcal{S^\perp}\|_{\textup{op}}\BI \right] \breve \Fx_t \\
    &\quad + [L_{\textup{b}}^\TRANS \pi_t \BI]   \breve \Fu_t^{app},
    \end{aligned}
\end{equation}
and cost in \eqref{eq:approximate-aux-cost}.
When the same control feedback gain is applied to the dynamics in \eqref{eq:aux-dynamics-with-special-control},  the close-loop dynamics (projected in the subspace $\mathcal{S^\perp}$) converges to the origin faster than the closed-loop dynamics for \eqref{eq:scalar-aux-dynamics},
since the {real parts of all the values in the spectrum of} following difference operator
\begin{equation*}
\begin{aligned}
	\Delta(t)\triangleq&\left[L_{\textup{a}} \BI +D_{\textup{a}}\FA_\mathcal{S^\perp}- D_{\textup{b}} L_{\textup{b}}^\TRANS  \pi_t \FB_\mathcal{S^\perp}\right]\\
	&- \left[L_{\textup{a}} \BI +\|\FA_\mathcal{S^\perp}\|_{\textup{op}} D_{\textup{a}}\BI +D_{\textup{b}}L_{\textup{b}}^\TRANS \pi_t \|\FB_\mathcal{S^\perp}\|_{\textup{op}}\BI \right]\\
	 =& D_{\textup{a}}(\FA_\mathcal{S^\perp}-\|\FA_\mathcal{S^\perp}\|_{\textup{op}} \BI)- D_{\textup{b}}L_{\textup{b}}^\TRANS \pi_t(\FB_\mathcal{S^\perp}+\|\FB_\mathcal{S^\perp}\|_\textup{op} \BI)\\
\end{aligned}
\end{equation*}
are always non-positive for all $t\in[0,T]$.
\begin{figure*}[!t]
\centering
  \subfloat[{Graphon approximate control and optimal control}]{\includegraphics[width=9cm,trim = {1.5cm 0 1cm 0}, clip]{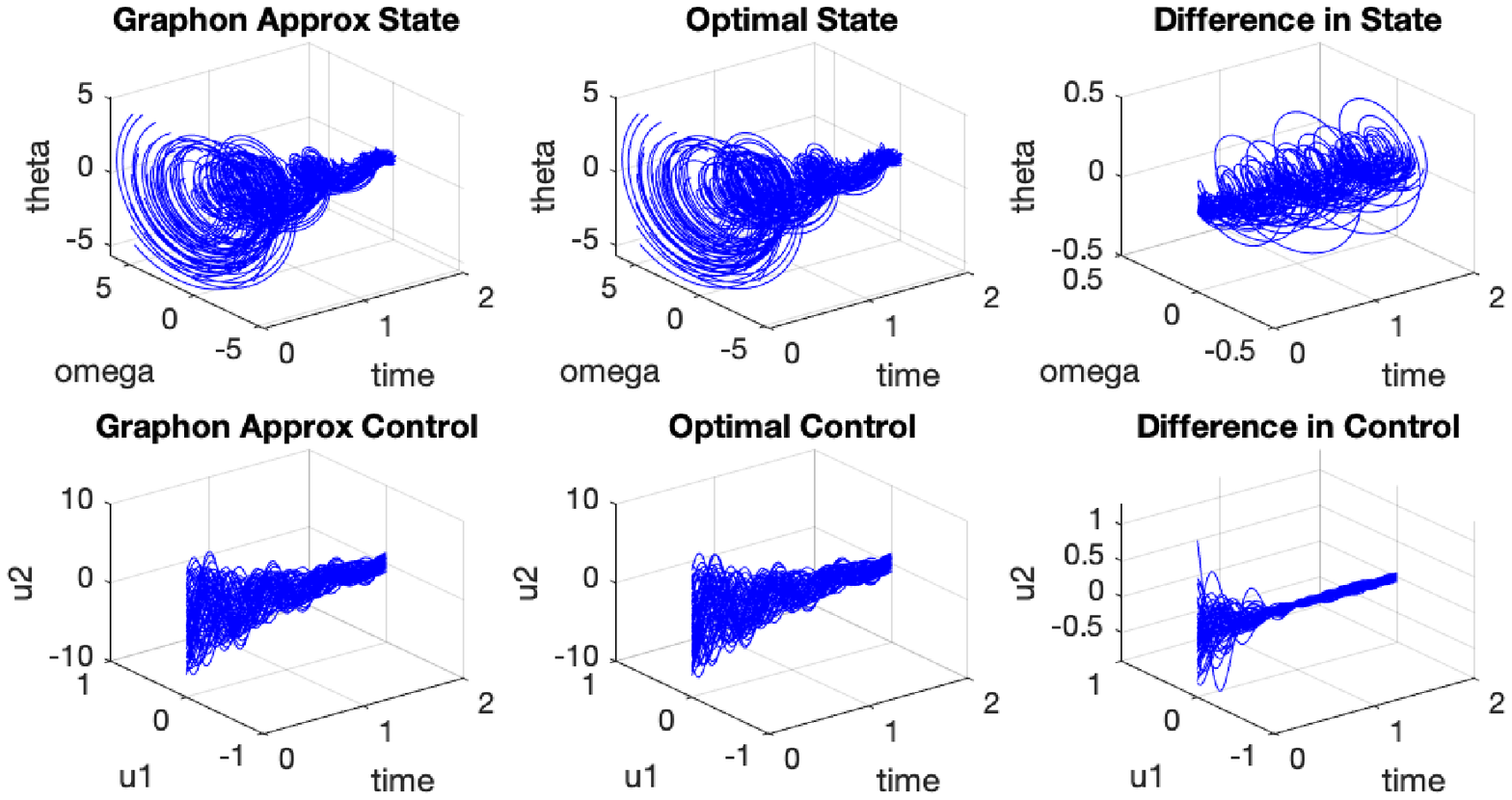}}
  \subfloat[{Projection-based approximate control and optimal control}]{\includegraphics[width=9cm,trim = {1.5cm 0 1cm 0}, clip]{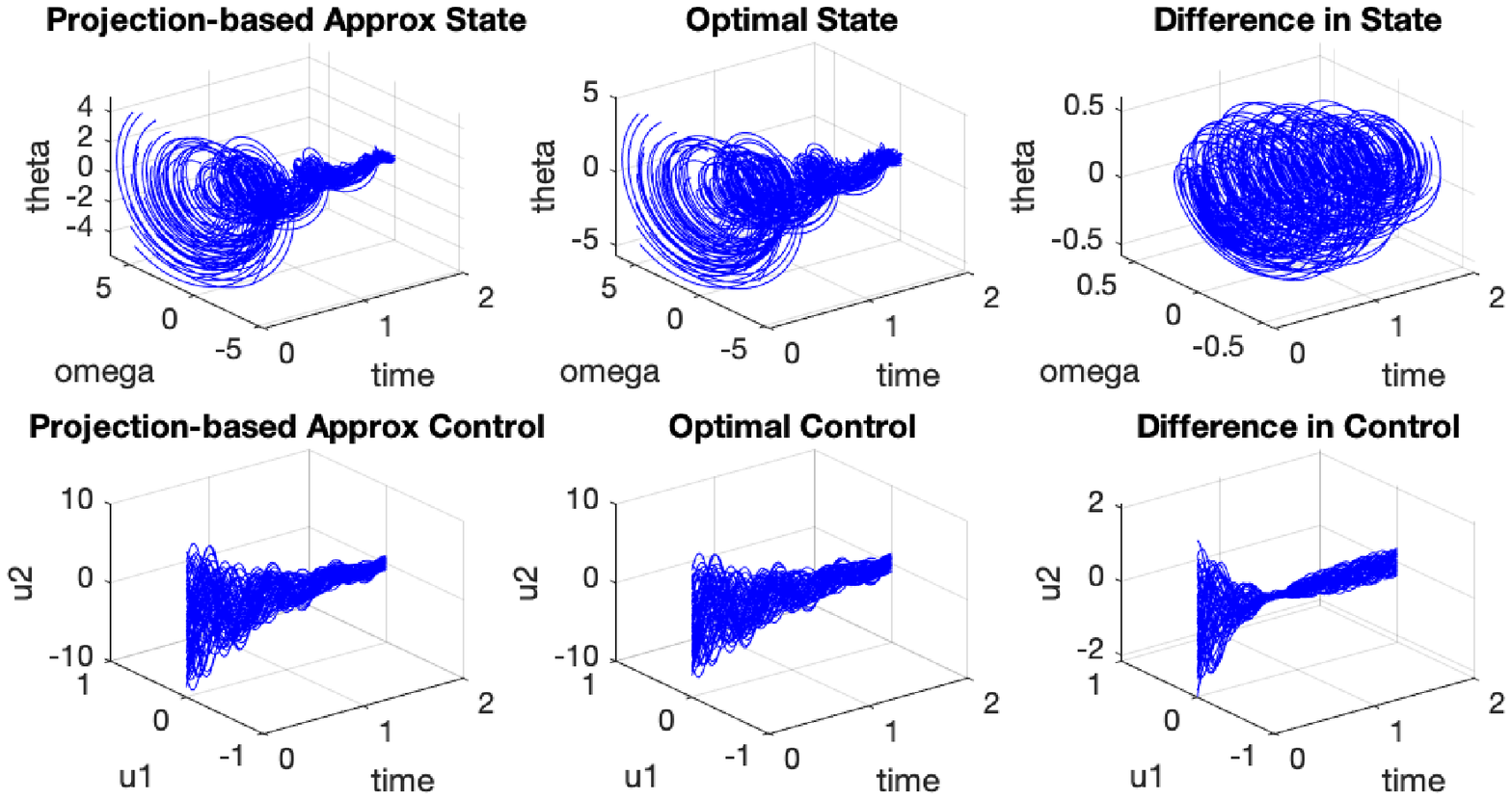}}

  \caption{{Graphon approximate control following \cite{ShuangPeterTAC18} and projection-based approximate control following Section \ref{sec:approxmiate-control} are applied to a network of 60 harmonic oscillators. 
    The second columns in (a) and (b) denote the  trajectories of state and control under the optimal control, and the last columns in (a) and (b) present the differences between trajectories under approximate control and those under optimal control.   
    In (a), the $L^2([0,T];\BR^{120})$ norm of the state difference in the last column is $5.87\%$ of that of the optimal state trajectory in the second column, and the approximate control cost is $0.461\%$ higher than the optimal cost. In (b), the $L^2([0,T];\BR^{120})$ norm of the state difference is $14.82\%$ of that of the optimal state trajectory, and the approximate control cost is $3.356\%$ higher than the optimal cost.    }}
    \label{fig:coupled-ocillators}
\end{figure*}

When (A5)-(ii) also holds, this approximate control implementation recovers the exact optimal control in Theorem \ref{thm:main-result}. {Furthermore, an approximate collaborative control similar to that in Corollary \ref{cor:local collaborative} may be generated by simply replacing $\breve \Fu^{o\gamma}$ there with the approximate auxiliary control $\breve \Fu^{app~\gamma}$ in \eqref{eq:approx-opt-control-aux} for all $\gamma \in [0,1]$.}

A numerical illustration is shown in Fig. \ref{fig:approximate-control-simulation}, where the underlying network (or graphon) couplings contain uncertainties and are generated from a stochastic block (graphon) model  as in Fig. \ref{fig:graphs-graphon}. These networks can be well approximated by low-rank models 
and there is usually a clear spectral gap between the most significant eigenvalues and the rest.  The size of the network in the illustative example is $120$. Based on low-rank approximations, the approximate control is generated and implemented. 
{The parameters in the simulation are: $L_{\textup{a}} = 2,   L_{\textup{b}}= 1.2,  L_{\textup{q}_\textup{T}} = 2$, $ L_{\textup{q}}= D_{\textup{a}}=D_{\textup{b}}=D_{\textup{q}}=D_{\textup{q}_\textup{T}}=1$.  The underlying network (or graphon) couplings $\FA$ and $\FB$ are generated from the stochastic block model in Fig. \ref{fig:graphs-graphon} and $\FQ=\FQ_T = \FA$. The subspace corresponding to the three most significant eigenvalues of $\FA$ is considered. The  operator norms $ \|\FA_{\mathcal{S}^\perp}\|_{\textup{op}}= \|\FQ_{\mathcal{S}^\perp}\|_{\textup{op}}=\|\FQ_{T\mathcal{S}^\perp}\|_{\textup{op}}=0.058$ and $    \|\FB_{\mathcal{S}^\perp}\|_{\textup{op}}= 0.076$ are assumed available for the computation of the control law. {
The initial conditions for all agents are uniformly sampled from $[-5,+5]$. 
Each of the normalized eigenvectors associated with the three most significant eigenvalues of graphs  in Fig.~\ref{fig:graphs-graphon} contains roughly 3 block structures. The projected states in each direction 
correspond roughly to the weighted sums of the block averages of initial states. Therefore, in this simulation example, the initial values of the subspace projected states are often small compared to the initial values of the actual states and the auxiliary states (see Fig. \ref{fig:approximate-control-simulation}). 
}
}

\section{{Regulating Coupled Harmonic Oscillators}}\label{sec:harmonic}
Consider a very large-scale network of coupled harmonic oscillators
\begin{equation}
  \dot{x}_t^i = \alpha \begin{bmatrix} 0 &1 \\-1 &0 \end{bmatrix} x_t^i + \frac{1}{N}\sum_{j=1}^N a_{ij} x_t^j + \beta \begin{bmatrix}0 &0 \\0 & 1 \end{bmatrix} u^i_t,
\end{equation}
where $\alpha, \beta \in \BR_+$,  $x_t^i, u_t^i  \in \BR^2$. Here $ \alpha $ represents the natural frequency of the harmonic oscillators,   $x_t^i \triangleq [\theta_t^i, \omega_t^i]^\TRANS$ is the state (which  may represent, for instance, location and velocity) and the second component of $u_t^i$ represents the  input force of the $i$th harmonic oscillator. 
The objective is to design a control law that minimizes the following cost with network couplings:
\begin{equation*}
\begin{aligned}
  J(u) =& 
   \frac1N\sum_{i=1}^N \Big \{\int_0^T\Big[(x_t^i - \eta z_t^i )^\TRANS Q(x_t^i - \eta z_t^i)  + (u_t^i)^\TRANS R u_t^i\Big] dt
 \\& \qquad  + (x_T^i - \eta z_T^i)^\TRANS Q_T(x_T^i - \eta z_T^i)\Big \},
\end{aligned}
\end{equation*}
where $z_t^i =  \frac{1}{N}\sum_{j=1}^N a_{ij} x^j_t$,  $Q,Q_T\geq 0$ and $R>0$.  Denote 
$$
L_{\textup{a}}=\begin{bmatrix} 0 &\alpha \\-\alpha &0 \end{bmatrix}, ~~ L_{\textup{b}} = \begin{bmatrix} 0 &0 \\0 &\beta \end{bmatrix}.$$

Assume the underlying graph  lies in a sequence of graphs which converges to some graphon limit, as depicted  by the  sequence of graphs shown in Fig.~\ref{fig:graphs-graphon}. One can then formulate the limit graphon LQR problem for systems distributed on the underlying graph.
Adopting Assumptions (A1)-(A5), and based upon the subspace decompositions introduced above, the optimal control for the limit problem is given by %
\begin{equation}\label{eq:HMO-control-sol}
\begin{aligned}
   \Fu_t (\gamma) &= - L_{\textup{b}}^\TRANS \left[\breve \Pi_{t} \breve \Fx_t(\gamma) + \sum_{\ell =1}^d \Pi_{t}^\ell  \Fx_t^\ell(\gamma)\right] \\
\end{aligned}
\end{equation}
where $\gamma \in  [0,1]$ represents a agent in the network with state $\Fx_t(\gamma)\in \BR^2$ and control $\Fu_t (\gamma)\in \BR^2$,  $\breve \Pi$ and $\Pi^\ell$ are the solutions to the following matrix Riccati equations
\begin{equation}\label{eq:HMO-ex-sol}
\begin{aligned}
  &-\dot{ \breve \Pi}_t = L_{\textup{a}}^\TRANS \breve \Pi_t + \breve \Pi_t L_{\textup{a}} -  \breve \Pi_t L_{\textup{b}} L_{\textup{b}}^\TRANS \breve \Pi_t + Q, \\
  & -\dot{\Pi}^\ell_t = (L_{\textup{a}} + \lambda_\ell I)^\TRANS  \Pi_t^\ell  \\& \qquad ~+   \Pi_t^\ell (L_{\textup{a}} + \lambda_\ell I)^\TRANS- \Pi_t^\ell L_{\textup{b}} L_{\textup{b}}^\TRANS \Pi_t^\ell + (1-\eta\lambda_\ell)^2Q, \\
  &\quad \breve \Pi_T = Q_T,\quad  \Pi_T^\ell = (1-\eta\lambda_\ell)^2Q_T,\quad  1 \leq \ell \leq d.
\end{aligned}
\end{equation}

Two alternatives for generating control laws are possible: 
\begin{itemize}
	\item[(i)] Following the graphon control methodology in \cite{ShuangPeterTAC18}, the limit graphon control in \eqref{eq:HMO-control-sol} can then be applied to systems on networks of arbitrary sizes in the convergence sequence;
	\item[(ii)] The projection-based approximate control solution  given by \eqref{eq:approx-control-total} in Section \ref{sec:approxmiate-control}  provides an alternative to generate an approximate control for a finite network system.
\end{itemize}

{Numerical simulations based on (i) the graphon control methodology in \cite{ShuangPeterTAC18}, and (ii) the projection-based approximate control implementation in Section \ref{sec:approxmiate-control}, are presented in Fig.~\ref{fig:coupled-ocillators}.} 
For these simulations, we set the following parameters: 
$$  \alpha =10;~\beta =1.5;~Q=I;~Q_T=2I;~R=I;~\eta = 3; ~N =60. $$
{The time interval $[0,T]$ with $T=2$ is discretized into $200$ time steps.} The initial conditions for all agents are uniformly sampled from $[-5,+5]$. 
The couplings are represented by a graph in a convergent sequence generated  from the stochastic block model as in Fig.~\ref{fig:graphs-graphon}. Note that the rank of the limit graphon (i.e., the step function graphon that corresponds to the block matrix) for the particular example is $3$. 
{The projection-based approximate control method employs  projections into the three most significant eigendirections. In addition, the residual operators used in the projection-based approximate control in  Fig. \ref{fig:coupled-ocillators}(b) are $\FA_{\SBS^\perp}$,  $\FB_{\SBS^\perp}$, $\FQ_{\SBS^\perp} =\FQ_{T\SBS^\perp}= (\BI-\eta \FA_{\SBS^\perp})^2-\BI$ with  $\|\FA_{\SBS^\perp}\|_{\textup{op}} =  0.077$, $\|\FB_{\SBS^\perp}\|_{\textup{op}}=0$, and $\|\FQ_{\SBS^\perp} \|_{\textup{op}}=\|\FQ_{T\SBS^\perp}\|_{\textup{op}}=0.472$.}

{Each of the approximate solutions involves solving one $2\times 2$ Riccati equation and one $6\times 6$ Riccati equation, which can be further decomposed into $4$ decoupled Riccati equations of dimension $2\times 2$ as in \eqref{eq:HMO-ex-sol}. The corresponding actual computation is more than 29 times faster than solving $120\times 120$ dimensional Riccati equation required by a direct solution in the  simulation. The computation saving becomes more significant for network systems with larger sizes in the convergence sequence. }

\section{Discussion}
LQR problems on VLSNs of arbitrary sizes can be approximately solved by low-complexity methods based on subspace decompositions of graphon dynamical systems in two ways:
\begin{enumerate}
    \item[(i)] Following the graphon control methodology proposed in \cite{ShuangPeterTAC18}, the control law for the limit graphon system is employed to generate approximate controls for network systems that are in a sequence that converges to the limit system \cite{ShuangPeterTAC18} as illustrated in Fig. \ref{fig:coupled-ocillators}(a);
    \item[(ii)]  Any finite network LQR problem interpreted as a special case of graphon LQR problem can be solved  via a representation of the underlying graphons by step functions with $N \times N$ blocks where $N$ is the size of the network following Section \ref{subsec:rel-finite-with-graphon}. This is illustrated in Fig. \ref{fig:coupled-ocillators}(b) based on the projection-based approximate control in Section \ref{sec:approxmiate-control}. 
\end{enumerate}
Each of the alternative methods above involves solving two decoupled LQR problems where one requires solving a Riccati equation of dimension ${nd\times nd}$ and the other requires solving a Riccati equation of dimension $n\times n$. 
As a comparison, a direct approach to the solution of LQR problems on networks with $N$ agents requires solving a Riccati equation of dimension $nN \times nN$. Since $N\geq d$ and in some cases $N\gg d$, the solution method may lead to significant computational savings depending upon the underlying network property.  Furthermore, the method proposed here is potentially scalable since its complexity does not directly depend on the size of the network, {as illustrated by the harmonic oscillator example in Section~\ref{sec:harmonic}.} 

\section{Conclusion}
{This article proposes solutions to a class of graphon LQR problems based on invariant subspace decompositions where the couplings appear in states, controls and cost, and these couplings may be represented by different graphons.} Future directions of this line of research include the following: 1) the case with heterogeneous parameters for local dynamics, 2) problems with nonlinear local dynamics,  3) the study of receding horizon control with quadratic cost based on graphon approximations and 4) the relation between graphon dynamical systems and systems described by partial differential equations.

\bibliographystyle{IEEEtran}
\bibliography{mybib}

\end{document}